\documentclass[11pt]{article}

\usepackage{amsmath,amssymb,amsthm,hyperref}
\usepackage{mathtools}
\usepackage{enumitem}
\usepackage{algorithm}
\usepackage{algpseudocode}
\usepackage{braket}
\usepackage{comment}
\usepackage{setspace}
\usepackage{moresize}

\newcommand{\E}{E}

\newcommand{\R}{\mathbb{R}}

\newtheorem{assumption}{Assumption}
\newtheorem{remark}{Remark}
\newtheorem{theorem}{Theorem}
\newtheorem{corollary}{Corollary}

\newtheorem{lemma}{Lemma}

\usepackage[margin=1in]{geometry}

\begin{document}

\title{Quantum Estimation of Delay Tail Probabilities in Scheduling and Load Balancing}

\author{R. Srikant\\ECE, CSL and NCSA\\UIUC\\
rsrikant@illinois.edu}

\maketitle


\begin{abstract}
Estimating delay tail probabilities in scheduling and load balancing systems is a critical but computationally prohibitive task due to the rarity of violation events. Quantum Amplitude Estimation (QAE) offers a generic quadratic reduction in sample complexity ($O(p^{-1/2})$ vs $O(p^{-1})$), but applying it to steady-state queueing networks is challenging: classical simulations involve unbounded state spaces and random regeneration cycles, whereas quantum circuits require fixed depth and finite registers.

In this paper, we develop a framework for quantum simulation of delay tail probabilities based on truncated regenerative simulation. We show that regenerative rare-event estimators can be reformulated as deterministic, reversible functions of finite random seeds by truncating regeneration cycles. To control the resulting bias, we use Lyapunov drift and concentration arguments to derive exponential tail bounds on regeneration times. This allows the truncation horizon—and hence the quantum circuit depth—to be chosen such that the bias is provably negligible compared to the statistical error. The proposed framework enables quantum estimation in models with countably infinite state spaces, avoiding the challenge of determining the sufficient mixing time required for direct finite-horizon simulation. We provide bounds on qubit and circuit complexity for a GI/GI/1 queue, a wireless network under MaxWeight scheduling, and a multi-server system with Join-the-Shortest-Queue (JSQ) routing.
\end{abstract}


\maketitle

\section{Introduction}
The rigorous estimation of tail probabilities is a cornerstone of performance evaluation, critical for characterizing Quality-of-Service (QoS) in modern networking systems, whether ensuring Service Level Objectives (SLOs) in data centers or bounding latency in ultra-reliable low-latency communications (URLLC). In these settings, system operators are often interested in rare events—such as the probability that a packet delay exceeds a high threshold $d$—where target probabilities may be of the order of $10^{-6}$ or even as low as $10^{-9}$ for industrial motion control and 6G-ULLRC \cite{saad2019vision,popovski2022perspective,3gpp22804}.

Classical Monte Carlo (CMC) simulation is the standard tool for such estimation, yet it suffers from high computational costs in the rare-event regime. To estimate a probability $p$ with fixed relative error using CMC, the required number of samples scales as $O(p^{-1})$. For rare events, this cost is prohibitive. This bottleneck has driven decades of research into variance reduction techniques, such as Importance Sampling (IS), cross-entropy methods, and Multi-level Splitting~\cite{asmussen2007applied, asmussen2007applied, blanchet2008efficient, glasserman1999multilevel, glynn1994simulation, heidelberger1995fast}, which aim to reduce the sample complexity.

In parallel, the field of quantum computing has developed powerful algorithms for probabilistic estimation that offer asymptotic speedups over classical counterparts. Most notably, Quantum Amplitude Estimation (QAE)~\cite{brassard2002quantum} allows for the estimation of a probability $p$ with a query complexity scaling as $O(p^{-1/2})$. This quadratic speedup is generic: unlike variance reduction, it does not strictly require domain-specific modifications to the underlying stochastic process. For a tail probability in the range of $10^{-6}$ to $10^(-9)$, QAE could theoretically reduce the simulation burden from billions of samples to tens of thousands, offering a pathway to efficient certification of high-reliability systems. Of course, the practical benefit of this quadratic improvement depends critically on the cost of implementing the underlying coherent simulation oracle, which is the main focus of this paper.

However, bridging the gap between classical queueing simulation and quantum algorithms is non-trivial. Queueing simulations are typically defined over continuous time with potentially infinite state spaces and unbounded random variables (e.g., inter-arrival times). In contrast, quantum algorithms operate on finite-dimensional Hilbert spaces (qubits) using fixed-depth unitary circuits. Standard QAE requires the estimation problem to be cast as a unitary operator $\mathcal{A}$ that prepares a state $|\Psi\rangle=\sqrt{1-p}|\Psi_{0}\rangle|0\rangle+\sqrt{p}|\Psi_{1}\rangle|1\rangle$, where the probability of measuring $|1\rangle$ corresponds to the target $p$. Constructing such an operator for a queueing system requires two fundamental shifts in perspective:
\begin{enumerate}
    \item \textit{Deterministic, Reversible Dynamics:} Classical simulations often treat random number generation as an external call. In the quantum setting, the simulation must be expressed as a deterministic, reversible function of a finite input ``seed'' register in superposition.
    \item \textit{Finite Horizon Truncation:} While classical regenerative simulation relies on random cycle lengths $\tau$, a coherent quantum circuit must have a pre-determined, fixed gate depth. This necessitates truncating the simulation at a horizon $M$ and rigorously bounding the resulting bias.
\end{enumerate}

In this paper, we develop a framework to apply QAE to the estimation of tail probabilities in queueing networks. Our primary focus is on delay or response time tail probabilities, though the methodology extends naturally to queue lengths and other QoS metrics. We formulate the regenerative simulation of a queue not as a sequential stochastic process, but as a deterministic boolean function $f:\{0,1\}^{n}\rightarrow[0,1]$ mapping a random seed to a truncated estimator. We show that by combining this formulation with rigorous tail bounds derived from Lyapunov drift analysis~\cite{hajek1982hitting, meyn1993stability, meyn2009markov}, we can explicitly control the error introduced by circuit truncation.

A natural question is why one would not instead rely on classical variance-reduction techniques, such as importance sampling or splitting, to estimate rare-event probabilities. While these methods can be extremely effective when a suitable change of measure or level structure is available, their design is typically highly problem-specific and often relies on detailed large-deviations analysis or model-specific insight. For complex queueing networks or multi-dimensional state spaces, constructing provably efficient importance sampling schemes remains challenging or unresolved. Our focus in this paper is complementary. Rather than proposing new variance-reduction estimators, we show that the baseline regenerative estimator—once suitably truncated—can be implemented coherently and evaluated using QAE. This yields a quadratic reduction in the number of estimator evaluations required, without assuming the existence of an efficient importance sampling scheme.

The main contributions of this paper are as follows:
\begin{itemize}
    \item \textit{Regenerative simulation as a quantum oracle.} We show that regenerative rare-event simulation for queueing systems can be reformulated as a deterministic, finite-depth, reversible computation suitable for QAE. The key idea is to truncate regeneration cycles at a fixed horizon and to express the resulting single-cycle estimator as a function of a finite random seed, enabling coherent quantum evaluation despite infinite state spaces and random cycle lengths.
    \item \textit{Circuit depth from drift analysis.} We develop bounds on the error induced by truncating regeneration cycles using Lyapunov drift and concentration arguments. While classical tools available in the literature \cite{hajek1982hitting,meyn2009markov} show that such bounds exist, we employ these tools here to derive explicit exponential tails for regeneration times, and allow quantum circuit depth to be chosen so that truncation bias is provably negligible compared to the statistical error guaranteed by QAE. 
    \item \textit{Quantum complexity bounds for canonical queueing models.} We provide detailed qubit and circuit complexity bounds for implementing the proposed framework, accounting for reversible arithmetic and uncomputation overhead. The analysis is illustrated for three representative systems: a GI/GI/1 queue, a wireless network under MaxWeight scheduling, and a multi-server load-balancing system with Join-the-Shortest-Queue (JSQ) routing (employing Nummelin splitting to handle the continuous state space), demonstrating how model-specific drift properties translate into concrete quantum resource requirements.
\end{itemize}

A note on implementation: our analysis assumes a Fault-Tolerant Quantum Computing (FTQC) model due to the circuit depths required for long regeneration cycles. In the quantum setting, variance reduction may potentially be used to shorten effective regeneration cycles and thereby reduce circuit depth, so that techniques offering limited benefit in classical simulation may nonetheless be useful for enabling implementations on Noisy Intermediate-Scale Quantum (NISQ) devices. 

\subsection{The Challenge: Steady-State vs. Finite-Depth}
A fundamental conceptual difficulty in applying quantum algorithms to queueing networks lies in the mismatch between the infinite-horizon nature of steady-state metrics and the finite-depth requirement of quantum circuits. In classical discrete-event simulation, there are two primary methods to estimate steady-state quantities:
\begin{enumerate}
    \item \textit{Finite-Horizon (Time-Averaged) Simulation:} One simulates the system for a fixed, large time $T$ and averages the metric. This approach suffers from initial transient bias (or ``warm-up'' bias) because the system starts from an atypical state (usually empty). To reduce this bias, $T$ must be significantly larger than the mixing time of the system~\cite{asmussen2007applied}.
    \item \textit{Regenerative Simulation:} One breaks the trajectory into independent, identically distributed cycles defined by returns to a regeneration state (e.g., the empty system)~\cite{asmussen2007applied, AsmussenGlynn2007, glynn1994simulation, kang2007exploiting}. Steady-state expectations are then expressed as ratios of cycle expectations.
\end{enumerate}

Translating these to a quantum setting presents a dilemma. QAE requires the underlying boolean oracle to be a circuit of fixed, pre-determined depth $M$~\cite{brassard2002quantum}.
\begin{itemize}
    \item A \textit{Quantum Finite-Horizon} approach would require fixing $M$ large enough to reach steady-state. However, for high-reliability networks, the mixing time is not only potentially enormous but also difficult to bound rigorously. Underestimating the mixing time risks introducing undetected initialization bias.
    \item A \textit{Quantum Regenerative} approach is attractive because it eliminates the warm-up bias. However, the cycle length $\tau$ is a random variable with no hard upper bound, violating the fixed-depth constraint of the quantum circuit.
\end{itemize}

\subsection{Related Work}

\textit{Rare-Event Simulation.}
The estimation of rare-event probabilities in queueing systems has been extensively studied, with Importance Sampling (IS) and Splitting as the dominant approaches \cite{AsmussenGlynn2007}. IS accelerates rare-event observation via a change of measure and likelihood reweighting \cite{asmussen2007applied, glynn1994simulation}, but its effectiveness depends critically on selecting an appropriate tilt, often guided by large deviations analysis \cite{siegmund1976importance}. Splitting techniques \cite{villen2002restart} avoid likelihood ratios but require the design of suitable intermediate levels, which can be challenging in high-dimensional systems. Our focus is complementary: we enable QAE for the baseline regenerative estimator, while leaving variance reduction.

\textit{Quantum Amplitude Estimation (QAE).}
QAE, introduced by Brassard et al.~\cite{brassard2002quantum}, provides a quadratic improvement over classical sampling for estimating expectations. Several variants aim to reduce circuit depth or remove quantum phase estimation requirements \cite{suzuki2020amplitude, grinko2019iterative}. Montanaro~\cite{montanaro2015quantum} analyzed the use of QAE for Monte Carlo estimation of expectations. Our work specializes this framework to stochastic processes with regenerative structure, addressing challenges arising from unbounded state spaces and random regeneration times that are absent in static Monte Carlo settings.

\textit{Quantum Simulation of Stochastic Processes.}
Related work on quantum simulation of stochastic dynamics has largely focused on state preparation or sampling from stationary distributions. In contrast, our setting requires estimating probabilities of rare trajectories and implementing long-horizon dynamics in a reversible manner, connecting to classical results on reversible computation and time--space tradeoffs \cite{bennett1973, bennett1989}.

\textit{Quantum Queueing Simulations.}
Recent studies have explored quantum simulation of queueing models, primarily focusing on finite-capacity systems or bounded truncations of infinite dynamics. Peretz et al.~\cite{peretz2025quantum} developed coherent circuits for $M/G/1/K$ queues, where the Hilbert space is naturally bounded by the physical capacity $K$. For infinite $M/M/1$ systems, Koren and Peretz~\cite{koren2024dynamic} proposed a  technique where a queue threshold is fixed according to the traffic load; this ensures that the probability of the state trajectory exceeding the threshold is negligible relative to the standard simulation error tolerance. Our work targets more general \emph{infinite}-state queueing models with a specific focus on estimating \emph{rare-event delay tails}. We address this by combining regenerative simulation with rigorous truncation bounds derived from Lyapunov drift arguments.

\subsection{Organization of the Paper} 
The remainder of the paper is organized as follows. 
Section~\ref{sec:G-G-1} presents the core framework using a $GI/GI/1$ queue. 
Sections~\ref{sec:wireless} and~\ref{sec:load-balancing} extend the approach to wireless scheduling and multi-server load balancing, respectively. 
Section~\ref{sec:conc} provides concluding remarks, and the appendix presents some extensions and some details omitted in the main body of the paper due to space limitations.

\section{Waiting Time Tail Estimation in $GI/GI/1$ Queues}\label{sec:G-G-1}

In this section, we make the connection between classical regenerative simulation and quantum simulation in a simple infinite-state setting, serving as a template for more complex models in later sections. We consider a continuous-time $GI/GI/1$ queue with i.i.d.\ inter--arrival times
$\{A_k\}$ and i.i.d.\ service times $\{S_k\}$, satisfying
$E[S_1] < E[A_1]$.
Recall the Lindley recursion for the waiting time given by
\begin{equation}
\label{eq:lindley}
W_{n+1} = \bigl(W_n + S_n - A_{n+1}\bigr)^+,\qquad n\ge 0,
\end{equation}
where $W_n$ is the waiting time of the $n^{\rm th}$ arrival.
For a fixed threshold $d > 0$, our goal is to estimate the
delay tail probability
\[
p_d := P(W \ge d),
\]
where $W$ denotes the stationary waiting time.
Since delay (also called response time or sojourn time, which is the time between arrival and eventual departure of a packet or a job) is related to the waiting time in a simple manner, i.e., we just have to add the service time, it is a straightforward exercise to modify our results to estimate delay tail probability.
Therefore, we stick with waiting time tail estimation in this section.
A standard approach to estimating $p_d$ is based on regenerative simulation.
Let $\tau$ denote the length of a regeneration cycle, defined as number of arrivals between successive visits to $W_n=0$.
Let $R$ denote the
number of arrivals in a regeneration cycle whose waiting time exceeds $d$. By
renewal--reward theory,
\[
p_d = \frac{E[R]}{E[\tau]}.
\]

From a simulation perspective, estimating $E[\tau]$ is relatively
straightforward: standard Monte Carlo simulation yields accurate estimates
with negligible variance relative to the estimation of $E(R)$.
In contrast, estimating $E[R]$ is significantly more
challenging when $p_d$ is small, since $R$ is typically zero on most
regeneration cycles and nonzero only on rare cycles that contain extreme
waiting times.
Consequently, the computational bottleneck in latency tail estimation lies in
estimating $E[R]$, rather than $E[\tau]$.
This has motivated the use of
variance--reduction techniques tailored to rare--event simulation, which aim
to estimate $E[R]$ more efficiently than naive Monte Carlo.
Standard regenerative simulation requires us to average the number of arrivals over multiple regeneration cycles.
There are several issues with directly trying to perform such a simulation in a quantum computer using QAE:
\begin{itemize}
    \item The number of arrivals in a regeneration cycle (which we will call the length of the regeneration cycle from now on) is random without an upper bound.
Therefore, we have use truncation to bound the length of the regenerative cycle and account for the error due to truncation.
    \item To use QAE, we have to write the output of a single regeneration cycle (i.e., the number of packets whose waiting time was greater than or equal to $d$) as a deterministic function of a finite number of input bits whose output is also a finite number of bits.
There are multiple ways to do this: one way is to think of the simulation as it would be done on a classical computer where there is a finite-bit initial seed from which all ``randomness" is generated using a deterministic pseudo-random generator. Further, all operations performed within a regenerative cycle have to be thought of in terms of finite-bit binary numbers. We do not explicitly characterize finite-precision error, as the same issue arises in classical simulation and does not materially affect the comparison between classical and quantum implementations.
    \item Finally, all operations must be reversible for quantum implementation; this is one of the requirements of quantum computing. We will appeal to existing quantum computing literature to argue that all computations needed within a regeneration cycle can be implemented reversibly, with any ancilla (extra) bits required for reversibility uncomputed, i.e, set equal to zero.
\end{itemize}

\subsection{Truncation and Quantum Implementation}
\label{sec:trunc-norm-quantum}

We truncate the length of regeneration cycles at a fixed horizon $M$.
Let $\tau_M:=\min\{\tau,M\}$, and let $R_M:=\sum_{n=1}^{\tau_M}{1}\{W_n\ge d\}$
denote the truncated rare--event count.
QAE requires the random number whose expectation is to be estimated to lie in $[0,1].$ Therefore, we define the quantity
\[
Y := \frac{R_M}{M} \in [0,1].
\]
QAE would construct an estimate of $E[Y]$, from which one
recovers $E[R_M]=M\,E[Y]$.
The denominator $E[\tau]$ can
be estimated classically by standard Monte Carlo since it does not require rare--event
variance reduction;
the overall estimate is obtained by forming the ratio $E(R_M)/E(\tau)$ while characterizing the error due to truncation $R(R)-E(R_M)$.
We now write the computation of $Y$ for \emph{one} truncated cycle as a
deterministic function of a finite random seed.
Conceptually, the seed encodes
all randomness required to simulate up to $M$ steps of the queue dynamics.
While $S_n$ and $A_{n+1}$ are random variables, with a view towards quantum implementation, we think of them as being deterministically generated from a single seed, which is exactly how one generates ``random variables'' in a classical simulation.
Let $\omega\in\{0,1\}^m$ be an initial seed and let $\mathrm{PRNG}(\omega,j)$ be a
deterministic pseudorandom generator that outputs a fixed-length bitstring
interpretable as a uniform variate in $[0,1)$,
given a call index $j\in\{1,2,\dots,2M\}$.
We generate inter--arrival and service
times via deterministic maps $\mathcal{T}_A$ and $\mathcal{T}_S$ (e.g., inverse
CDF transforms, alias tables, or fixed discretizations):
\[
U_{2n-1} := \mathrm{PRNG}(\omega,2n-1),\quad A_{n+1} := \mathcal{T}_A(U_{2n-1}),
\]
\[
U_{2n} := \mathrm{PRNG}(\omega,2n),\quad S_n := \mathcal{T}_S(U_{2n}).
\]
Thus, truncation to $M$ arrivals uses exactly $2M$ PRNG calls.
\begin{algorithm}
\caption{\textsc{EvaluateTruncatedCycle}$(\omega)$ for $GI/GI/1$}
\label{alg:seed-cycle-lindley}
\begin{algorithmic}[1]
\State Initialize $W \gets 0$, $r \gets 0$, $n \gets 0$.
\While{$n < M$}
  \State $n \gets n+1$
  \State $U_{odd} \gets \mathrm{PRNG}(\omega,2n-1)$;
\ \ $A \gets \mathcal{T}_A(U_{odd})$.
  \State $U_{even} \gets \mathrm{PRNG}(\omega,2n)$; \ \ \ \ \ \ \ $S \gets \mathcal{T}_S(U_{even})$.
\State $W \gets (W + S - A)^+$ \Comment{Lindley recursion using $S_n, A_{n+1}$}
  \State $b \gets {1}_{W \ge d}$.
\State $r \gets r + b.$
  \If{$W=0$}
    \State \textbf{break} \Comment{regeneration (zero waiting time)}
  \EndIf
\EndWhile
\State \ $R_M \gets r$, \ \ $Y \gets R_M*M^{-1}$.
\State \Return $(Y)$.
\end{algorithmic}
\end{algorithm}

Algorithm~\ref{alg:seed-cycle-lindley} is a deterministic function of the seed
$\omega$ and the fixed truncation parameter $M$, hence defines a function
$f(\omega):=Y(\omega)\in[0,1]$ suitable for QAE.
\subsubsection{Qubit and circuit complexity}
\label{sec:resources-Gemini}

We summarize the quantum resources required to coherently evaluate the truncated
regenerative estimator $f(\omega)\in[0,1]$.
By Bennett’s reversibility theorem, any classical function $f$ can be implemented reversibly;
however, achieving linear time $O(M)$ requires a space-time trade-off for non-bijective operations like the Lindley recursion~\cite{bennett1973, bennett1989}.
To characterize the quantum circuit complexity, we decompose the algorithm into reversible primitives, accounting for the \textit{history bits} required to maintain reversibility.
\paragraph{Registers and qubit count.}
A rigorous reversible implementation requires the following registers:
\begin{itemize}
    \item a seed register of $m$ qubits holding $\omega$;
    \item a PRNG call counter register of $\lceil \log_2 M +1\rceil$ qubits;
    \item registers of $B_A$ and $B_S$ qubits for inter-arrival and service times, respectively;
   \item a \textit{waiting time state} register of $B_W$ qubits.
Since the waiting time accumulates service times over the horizon $M$, we set $B_W = \lceil \log_2 M \rceil + B_S + O(1)$ to accommodate the maximum expected value and prevent overflow;
    \item a \textit{history register} of $M \cdot B_W$ qubits.
    The Lindley update $W_{n+1} = \max(0, W_n + S_n - A_{n+1})$ is non-invertible whenever the result is clipped to zero. To enable uncomputation via Bennett's strategy~\cite{bennett1973}, the circuit must store the full ``undershoot'' value (the negative result before clipping) onto the history tape at each step $n$. This requires allocating a fresh block of $B_W$ qubits for every iteration $n \in \{1, \dots, M\}$;
    \item a counter register of $B_R=\lceil\log_2(M+1)\rceil$ qubits for the truncated rare-event count $R_M$;
    \item a fixed-precision output register of $B_Y$ qubits for the normalized value $Y = R_M \cdot M^{-1}$;
    \item temporary \textit{ancilla qubits} for reversible arithmetic, requiring $O(B_W)$ space~\cite{cuccaro2004adder};
    \item a control flag of $1$ qubit to implement \textit{logical early stopping}.
Since the quantum circuit must execute a fixed depth $M$, physical termination at a random regeneration time $\tau$ is impossible.
Instead, this flag conditions all update operations using standard controlled-gate constructions~\cite{vedral1996}: upon regeneration ($W=0$), the flag flips, forcing all subsequent operations to act as Identity gates.
This effectively ``freezes'' the state and prevents further evolution for the remaining $M-\tau$ steps.
\end{itemize}
Thus, the total qubit count scales as
\[
Q = m +  O(B_A + B_S + B_Y + M\log M).
\]
From now on, we will focus on complexity only in terms of $M$ since that is the primary factor we control in this paper.

\paragraph{Gate and depth complexity per cycle.}
Let $C_{\mathrm{step}}$ denote the circuit complexity of one iteration of
Algorithm~\ref{alg:seed-cycle-lindley}.
This cost includes the reversible implementation of addition, subtraction, constant multiplication, the indicator function calculation $1_{W > d}$, and the PRNG evaluation.
PRNG evaluation could incur a cost that is polynomial in $m,$ but as mentioned earlier, we focus on complexity in terms of $M.$ Using standard reversible arithmetic implementations~\cite{cuccaro2004adder,vedral1996}, the cost scales linearly with the register width:
\[
C_{\mathrm{step}} = O(B_W + B_A + B_S + B_R).
\]
Since both the queue width $B_W$ and the rare event counter $B_R$ scale logarithmically with the horizon $M$ (i.e., $B_W, B_R = O(\log M)$), the per-step cost simplifies to $C_{\mathrm{step}} = O(\log M + B_A + B_S)$.
Thus, one coherent evaluation of $f(\omega)$ requires
\[
C_f = O(M \cdot C_{\mathrm{step}}) = O(M (\log M + B_A + B_S))
\]
gates.
\paragraph{Overall quantum complexity.}
QAE estimates $E[f(\omega)]$ to additive error
$\varepsilon$ using $O(1/\varepsilon)$ coherent invocations~\cite{brassard2002quantum}.
The total circuit complexity scales as
\[
C_{\mathrm{QAE}} = O\!\left( \frac{M (\log M + B_A + B_S)}{\varepsilon} \right).
\]

\begin{remark}
\textit{Qubits vs.\ circuit complexity tradeoff with PRNG implementation:}
Generating the $2M$ samples from a single seed can be implemented either by storing all random blocks in the seed, or by invoking a PRNG indexed by a reversible counter.
The latter requires $O(\log M)$ additional qubits and incurs $O(M\log M)$ gate overhead.
The former avoids the counter but increases the seed length $m$ by $\Theta(M)$.
While both are standard, the history register required for the Lindley recursion ($M$ bits) means that both approaches share a linear dependence on $M$ for space.
\end{remark}

\begin{remark}
The response (sojourn) time of an arrival is given by $W_n + S_n$.
The estimator extends directly to response time tail probabilities by replacing the indicator ${1}\{W_n \ge d\}$ with ${1}\{W_n + S_n \ge d\}$.
As $W_n$ and $S_n$ are already available in registers, this modification adds only a single reversible addition and comparison per cycle, leaving the overall complexity bounds unchanged.
\end{remark}

\subsection{Quantifying the Error Due to Truncation}
\label{sec:error-trunc-norm}

To further understand the complexity of quantum implementation, we have to understand the error incurred by
truncating regeneration cycles at horizon $M$.

\begin{assumption}\label{ass:bounded increments}
Assume there exist known constants $A_{\max},S_{\max}>0$ such that, almost surely,
\[
0 \le A_n \le A_{\max},\qquad 0 \le S_n \le S_{\max}\qquad\text{for all }n.
\]
\end{assumption}

We start with the boundedness assumption to expose the main ideas in the clearest manner, but the results extend to more general settings as shown in Appendix~\ref{app:unbounded}.
Let
\begin{equation}
\label{eq:beta-def}
\beta := \frac{2\Delta^2}{(A_{\max}+S_{\max})^2},
\end{equation}
where
$
\Delta := E[A_1]-E[S_1] > 0.
$ is the stability slack.

\begin{theorem}
\label{thm:ER-error-bounded}

Fix a truncation horizon $M$. Let $\hat{\mu}_Q$
be an estimate of $E[Y]=E[R_M]/M$ returned by QAE
such that, with probability at least $1-\delta_Q$,
\begin{equation}
\label{eq:qae-err-bounded}
|\hat{\mu}_Q - E[Y]|
\le \varepsilon_Q.
\end{equation}
Then, with probability at least $1-\delta_Q$, the corresponding estimate
$\widehat{E[R]} := M\,\hat{\mu}_Q$ satisfies
\begin{equation}
\label{eq:ER-bound-bounded}
\Bigl|\,\widehat{E[R]} - E[R]\,\Bigr|
\;\le\;
M\varepsilon_Q
\;+\;
\frac{e^{-\beta M}}{1-e^{-\beta}},
\end{equation}
where $\beta$ is given by~\eqref{eq:beta-def}.
\end{theorem}

\begin{proof}
Define the i.i.d.\ increments
\[
X_i := S_i - A_{i+1},\qquad i\ge 0,
\]
so that $X_i\in[-A_{\max},S_{\max}]$ almost surely and
$E[X_i]=-(E[A_1]-E[S_1])=-\Delta$.

Let $\{V_t\}_{t\ge 0}$ be the associated unreflected random walk
\[
V_0 := 0,
\qquad
V_t := \sum_{i=0}^{t-1} X_i , \quad t\ge 1.
\]
For the reflected Lindley recursion $W_{t+1}=(W_t+X_t)^+$ with $W_0=0$, the
workload admits the standard representation
\begin{equation}
\label{eq:reflection-repr}
W_t = V_t - \min_{0\le k\le t} V_k
= \max_{0\le k\le t} \sum_{i=k}^{t-1} X_i ,
\qquad t\ge 0.
\end{equation}

Recall that $\tau:=\inf\{t\ge 1 : W_t=0\}$ is the first return time to zero.
On the event $\{\tau>t\}$ we have $W_1>0,\dots,W_t>0$.  By
\eqref{eq:reflection-repr}, the condition $W_m=0$ is equivalent to
$V_m=\min_{0\le k\le m}V_k$.  Therefore, $\{\tau>t\}$ implies that the
unreflected walk $V_m$ does not attain a new running minimum over
$m=1,\dots,t$, so that
\[
\min_{0\le k\le t} V_k = V_0 = 0.
\]
Consequently, $V_m>0$ for all $m=1,\dots,t$, and in particular $V_t>0$.
Thus,
\[
P(\tau>t)
\;\le\;
P(V_t>0)
=
P\!\left(\sum_{i=0}^{t-1} X_i \ge 0\right)
=
P\!\left(\sum_{i=0}^{t-1} (X_i-E[X_i]) \ge t\Delta\right).
\]

Applying Hoeffding's inequality to the bounded variables
$X_i-E[X_i]\in[-A_{\max}+\Delta,\,S_{\max}+\Delta]$, whose range length is at
most $A_{\max}+S_{\max}$, yields
\[
P(\tau>t)
\le
\exp\!\left(
-\frac{2t^2\Delta^2}{t(A_{\max}+S_{\max})^2}
\right)
=
\exp(-\beta t),
\]
where $\beta$ is defined in~\eqref{eq:beta-def}.

Since $R-R_M=0$ on $\{\tau\le M\}$ and $0\le R-R_M\le \tau-M$ on $\{\tau>M\}$,
it follows that
\[
0\le E[R]-E[R_M]
\le E[(\tau-M)^+]
= \sum_{t=M}^{\infty} P(\tau>t)
\le \sum_{t=M}^{\infty} e^{-\beta t}
= \frac{e^{-\beta M}}{1-e^{-\beta}}.
\]

If a quantum subroutine returns $\hat{\mu}_Q$ satisfying
\eqref{eq:qae-err-bounded}, then
\[
|M\hat{\mu}_Q - E[R_M]|
= M|\hat{\mu}_Q - E[Y]|
\le M\varepsilon_Q.
\]
Combining the bounds and applying the triangle inequality yields
\eqref{eq:ER-bound-bounded}.
\end{proof}
The Hoeffding can be conservative and is presented here only to avoid additional assumptions. With extra assumptions, one can get stronger bounds leading to the use of smaller $M$ in the quantum architecture.

\subsection{Certifying a waiting time tail bound: Main result for the $GI/GI/1$ queue}
\label{sec:certify-tail}

In this subsection, we assume that we are interested in certifying $p_d \le 10^{-k}$.
Conservatively, we then
assume that we would like to estimate $E[R]$ to an additive accuracy of $10^{-k-2}$ with success
probability at least $1-\alpha_Q$.
This enables certification of the tail bound
once $E[\tau]$ (or a lower bound on $E[\tau]$) is available
with high confidence from standard (non-rare-event) simulation, which can be done classically.
We will invoke the following well-known result in quantum computing \cite{brassard2002quantum, montanaro2015quantum}.
\begin{theorem}
\label{thm:qae}
Let $\Omega=\{0,1\}^m$ and let $\omega$ be drawn uniformly from $\Omega$.
Let $f:\Omega\to[0,1]$ be a function for which there exists a coherent
quantum circuit $U_f$ such that, given input $\ket{\omega}\ket{0}$,
the value $f(\omega)$ can be computed to sufficient precision in a work
register and used to implement a controlled rotation on an ancilla whose
squared amplitude equals $f(\omega)$, with all workspace uncomputed.

For any additive error $\varepsilon_Q>0$ and confidence
$\delta_Q\in(0,1)$, there exists a quantum algorithm that returns an estimate
$\hat\mu$ of
\[
\mu := E[f(\omega)]
\]
such that
\[
P\!\bigl(|\hat\mu-\mu|\le \varepsilon_Q\bigr)\;\ge\;1-\delta_Q,
\]
using
\[
O\!\left(\frac{1}{\varepsilon_Q}\log\frac{1}{\delta_Q}\right)
\]
coherent invocations of $U_f$ and $U_f^\dagger$.

The guarantee holds for any $f$ bounded in $[0,1]$, with no dependence on the
variance of $f$.
\end{theorem}

We now state the main result for the $GI/GI/1$ queue.
\begin{theorem}
\label{thm:choose-MC-for-ER}
Fix an integer $k\ge 1$ and a target failure probability
$\alpha_Q\in(0,1)$.
Set the target absolute accuracy
\[
\varepsilon_{\mathrm{tot}} := 10^{-k-2}.
\]
Assume the hypotheses of Theorem~\ref{thm:ER-error-bounded}
(bounded increments), so that the truncation error bound
\eqref{eq:ER-bound-bounded} holds.

Choose the truncation horizon as
\begin{equation}
\label{eq:choose-MC}
M := \left\lceil \frac{1}{\beta}\log\frac{4}{\beta\,\varepsilon_{\mathrm{tot}}}
\right\rceil.
\end{equation}
Run QAE on
\[
f(\omega)=Y(\omega)\in[0,1],
\]
with confidence $\delta_Q:=\alpha_Q$ and accuracy
\begin{equation}
\label{eq:choose-epsQ}
\varepsilon_Q := \frac{\varepsilon_{\mathrm{tot}}}{2M}.
\end{equation}
Then the resulting estimator $\widehat{E[R]}:=M\hat\mu$ satisfies
\[
P\!\left(
\bigl|\widehat{E[R]}-E[R]\bigr|\le \varepsilon_{\mathrm{tot}}
\right)\;\ge\;1-\alpha_Q.
\]

Furthermore, let $m$ be the number of qubits in the seed register (so that
$\Omega=\{0,1\}^m$), and suppose the coherent evaluation circuit for
$f(\omega)$ implements Algorithm~\ref{alg:seed-cycle-lindley} using a
pseudorandom generator indexed by $j\in\{1,\dots,2M\}$.
Then:

\begin{enumerate}
\item \textit{Qubit complexity.}
There exists an implementation using
\begin{equation}
\label{eq:qubits-cert}
Q
\;=\;
O(M)
\;+\;
O\!\left(\log\frac{1}{\varepsilon_{\mathrm{tot}}}\right)
\end{equation}
qubits, where the $O(\log(1/\varepsilon_{\mathrm{tot}}))$ term accounts for
fixed-precision arithmetic sufficient to ensure additive accuracy
$\varepsilon_{\mathrm{tot}}$ in the normalized output.

\item \textit{Circuit complexity.}
Let $T_f$ denote the circuit complexity of one coherent evaluation of
$f(\omega)$ (including PRNG calls, Lindley recursion updates, comparisons,
and uncomputation).
QAE yields an overall circuit complexity
\begin{equation}
\label{eq:gates-cert}
T_{\mathrm{QAE}}
\;=\;
O\!\left(
\frac{1}{\varepsilon_Q}\log\frac{1}{\alpha_Q}
\right)\,T_f
\;=\;
O\!\left(
\frac{M}{\varepsilon_{\mathrm{tot}}}\log\frac{1}{\alpha_Q}
\right)\,T_f,
\end{equation}
where we used~\eqref{eq:choose-epsQ}.

Under the seed-based simulation model with $2M$ PRNG calls and a reversible
counter, one may take
\begin{equation}
\label{eq:tf-bound}
T_f
\;=\;
O\!\left(
M\Bigl(\log M + \log\frac{1}{\varepsilon_{\mathrm{tot}}}\Bigr)
\right),
\end{equation}
up to polylogarithmic factors, yielding the bound
\begin{equation}
\label{eq:gates-cert-explicit}
T_{\mathrm{QAE}}
\;=\;
O\!\left(
\frac{M^2}{\varepsilon_{\mathrm{tot}}}
\Bigl(\log M + \log\frac{1}{\varepsilon_{\mathrm{tot}}}\Bigr)
\log\frac{1}{\alpha_Q}
\right),
\end{equation}
again up to polylogarithmic factors.
\end{enumerate}
\end{theorem}

\begin{proof}
By Theorem~\ref{thm:ER-error-bounded}, with probability at least $1-\delta_Q$,
\begin{equation}
\label{eq:split-error}
\bigl|\widehat{E[R]}-E[R]\bigr|
\;\le\;
M\varepsilon_Q
\;+\;
\frac{e^{-\beta M}}{1-e^{-\beta}}.
\end{equation}
We set $\delta_Q:=\alpha_Q$ and $\varepsilon_Q:=\varepsilon_{\mathrm{tot}}/(2M)$,
so that $M\varepsilon_Q=\varepsilon_{\mathrm{tot}}/2$.
Using $1-e^{-\beta}\ge \beta/(1+\beta)$ (equivalently $e^{\beta}\ge 1+\beta$),
we obtain
\[
\frac{e^{-\beta M}}{1-e^{-\beta}}
\;\le\;
\frac{1+\beta}{\beta}e^{-\beta M}
\;\le\;
\begin{cases}
\frac{2}{\beta}e^{-\beta M}, & 0<\beta\le 1,\\[4pt]
2e^{-\beta M}, & \beta\ge 1.
\end{cases}
\]
Thus, with
\[
M := \left\lceil \frac{1}{\beta}\log\frac{4}{\beta\,\varepsilon_{\mathrm{tot}}}\right\rceil,
\]
we have $e^{-\beta M}\le \beta\varepsilon_{\mathrm{tot}}/4$, so the truncation term
satisfies $\frac{e^{-\beta M}}{1-e^{-\beta}}\le \varepsilon_{\mathrm{tot}}/2$
in both cases:
for $0<\beta\le 1$,
$\frac{2}{\beta}e^{-\beta M}\le \varepsilon_{\mathrm{tot}}/2$,
and for $\beta\ge 1$,
$2e^{-\beta M}\le 2(\beta\varepsilon_{\mathrm{tot}}/4)\le \varepsilon_{\mathrm{tot}}/2$.
Substituting into~\eqref{eq:split-error} yields the claim with probability
at least $1-\alpha_Q$.
The qubit and circuit complexity statements follow from
Theorem~\ref{thm:qae} and the stated bound on $T_f$.
\end{proof}

\begin{remark}
A potential concern is the scaling of quantum resources with the truncation horizon $M$. Recovering the expectation $E[R_M]$ from the normalized estimator $Y=R_M/M \in [0,1]$ requires QAE precision scaling as $1/M$. Combined with the $O(M)$ gate depth per step, the total circuit complexity scales as $\tilde{O}(M^2)$, whereas classical simulation scales linearly with the average cycle length $E[\tau]$.

However, this polynomial overhead in $M$ is outweighed by the asymptotic advantage in sampling rare events. The horizon $M$ is determined by system stability (drift) and grows only logarithmically with the inverse error. The principal bottleneck is the rarity of the event $p_d$: classical simulation complexity scales as $O(p_d^{-1})$, while QAE scales as $O(p_d^{-1/2})$. For high-reliability targets (e.g., $p_d \approx 10^{-9}$), this quadratic reduction in sample complexity dominates the algebraic cost of truncation, particularly in light-traffic regimes where regeneration cycles remain short.
\end{remark}

\section{Latency-Tail Estimation in Wireless Networks}\label{sec:wireless}

We consider a single wireless base station serving $K$ receivers over slotted time
$t=0,1,2,\dots$.
Receiver $i$ maintains a FIFO queue of packets, and departures occur first,
followed by arrivals in each time slot.
For each queue $i\in\{1,\dots,K\}$, let $A_i(t)$ denote the number
of exogenous packet arrivals to queue $i$ during slot $t$.
The processes $\{A_i(t)\}_{t\ge 0}$ are i.i.d.\ over time, mutually independent across
queues, and satisfy
$
0 \le A_i(t) \le A_{\max}.
$
Let $\lambda_i := \E[A_i(0)]$.
If queue $i$ is scheduled in slot $t$, the number of packets that can be served is
$C_i(t)$.
The channel state vector $\mu(t)=(\mu_1(t),\dots,\mu_K(t))$ is i.i.d.\ over time,
independent of all arrivals, and satisfies
$
0 \le \mu_i(t) \le \mu_{\max}.
$

Let $Q_i(t)$ be the queue length of queue $i$ at the \emph{start} of slot $t$.
In each slot, at most one queue is scheduled.
Let $S(t)\in\{1,\dots,K\}$ denote the scheduled queue.
Departures are given by
\[
D_i(t)
=
1_{S(t)=i}\min\{Q_i(t),\,\mu_i(t)\},
\]
and the queue dynamics are
\[
Q_i(t+1) = Q_i(t) - D_i(t) + A_i(t).
\]

At the start of each slot $t$, the scheduler observes $(Q(t),\mu(t))$ and selects
\[
S(t)\in \arg\max_{1\le i\le K} Q_i(t)\,\mu_i(t),
\]
with a fixed deterministic tie-breaking rule.
Assume that the arrival rate vector $\lambda=(\lambda_1,\dots,\lambda_K)$ lies
strictly inside the stability region of the MaxWeight policy under the given
channel law, so that the Markov chain $\{Q(t)\}$ is positive recurrent and admits
a stationary distribution \cite{srikant2014communication}.
Each packet is time-stamped upon arrival.
Under FIFO service, the delay of a packet is defined as the number of slots
between its arrival and its departure.
We take the all-zero queue-length vector as the regeneration state and define
$\tau$ as the first return time to this state.
Let $I\subseteq\{1,\dots,K\}$ denote a subset of queues for which we seek a
delay guarantee of $d$ slots.
During a regeneration cycle, define:
\begin{itemize}
\item $N$: the total number of packet arrivals to queues $i\in I$ during the cycle,
\item $J(d)$: the number of those arrivals whose realized delays are greater than
or equal to $d$.
\end{itemize}

By standard regenerative arguments, the stationary probability that a
typical arrival to the set of queues $I$ experiences delay at least $d$ is given by
\[
P(\text{delay} \ge d)
=
\frac{\E[J(d)]}{\E[N]}.
\]

As we have seen the case of the $GI/GI/1$ queue, for quantum computing purposes, it is possible to define a function which takes a seed for pseudo-random number generation as input and outputs a single-cycle estimate of the number of delay violation events.
As before, we truncate a
regeneration cycle at a deterministic horizon $M$.
Starting from a regeneration time, simulation proceeds until time
$
\min\{\tau,\,M\}.
$
Only information available up to the truncation time is used to compute the
cycle statistics.
In a multi-queue wireless network, several notions of delay tail probabilities
are possible.
We focus on the packet-level delay tail defined above;
the results of the
paper can be adapted in a straightforward manner to alternative definitions.

Note that we have defined the regeneration event as the return of the queue-length vector to the all-zero state. Alternative regeneration constructions could be used to increase the frequency of regeneration. However, in the light-traffic regimes of interest in this paper, the all-zero state is sufficient. If desired, our analysis can be adapted in a straightforward manner to alternative regeneration events.

\subsection{A Tail Bound for the Regeneration time}
\label{subsec:maxweight-regeneration}

Since we take the all-zero state as the regeneration state, with $Q(0)=0,$ the regeneration time
\begin{equation}
\label{eq:tau-def}
\tau := \inf\{t\ge 1:\;
Q(t)=0\},
\end{equation}
i.e., the first return of the queue-length vector to the all-zero state after time $0$.
The only significant difference between the $GI/GI/1$ case is the model-dependent tail behavior of $\tau$, which controls the truncation error.
We therefore focus on deriving an exponential tail bound for $\tau$.

First, we recall that, for a wireless network operating under the MaxWeight policy, there exists a Lyapunov function $L\geq 0$ such that it has bounded one-step increments
\begin{equation}
\label{eq:bounded-increments}
|L(Q(t+1))-L(Q(t))|\le \nu \quad \mathrm{a.s.} \quad\forall t\geq 0
\end{equation}
and here exists $\varepsilon>0$ and $B>0$ such that the drift satisfies
\begin{equation}
\label{eq:negative-drift}
E\!\left[L(Q(t+1))-L(Q(t)) \mid Q(t)=q\right] \le -\varepsilon,
\qquad \forall 
q\notin C.
\end{equation}
where 
\[
C := \{q; L(q)\le B\}
\]
be a bounded set.

For MaxWeight scheduling, drift conditions of this type (with an appropriate choice of $L$ and $C$)
follow from standard throughput-optimality arguments when the arrival rate vector lies strictly inside
the stability region;
see, e.g., \cite{tassiulas1992stability} and the modification in \cite{eryilmaz2012asymptoticallytight} to satisfy the one-step bounded increment condition. Before we state the next theorem, we make the following assumption.

We will require the following finite-time emptying condition.
\begin{assumption}[Uniform finite-time emptying]
\label{ass:finite-emptying}
There exist integers $m\ge1$ and $p\in(0,1]$ such that
\begin{equation}
\label{eq:finite-emptying}
\inf_{q\in C} \mathbb{P}_q(\tau \le m) \ge p.
\end{equation}
\end{assumption}
This assumption is satisfied if, for example, there is non-zero probability of zero arrivals to each queue.

\begin{theorem}[Exponential tail of regeneration time]
\label{thm:maxweight-tail}
Consider the MaxWeight scheduling system described above.
Suppose that the drift conditions \eqref{eq:bounded-increments}--\eqref{eq:negative-drift} and Assumption~\ref{ass:finite-emptying} hold. Let $\kappa=\epsilon^2/(2\nu^2)$ and $\theta=\epsilon/\nu^2.$
Then the regeneration time $\tau$ satisfies
\[
P_0(\tau > t) \le \tilde{C}_\tau e^{-\eta t}, \qquad \forall t \ge 0,
\]
where the decay rate $\eta =\min(\kappa,\eta^*)/2$, $\eta^*$ is the unique solution satisfying
\[
(1-p)\,e^{\eta m}\,\exp\!\left(\frac{\eta}{\kappa}\theta m\nu\right) = 1,
\]
and $\tilde{C}_\tau := E_0[e^{\eta \tau}]$ is a finite constant bounded by \eqref{eq:exp-moment-tau}.
\end{theorem}

\begin{proof}

\textit{Step 1: Exponential return to a bounded set.}
Define the hitting time of $C$,
\[
\tau_C := \inf\{t\ge 0:\; Q(t)\in C\}.
\]
Under \eqref{eq:bounded-increments}--\eqref{eq:negative-drift}, a supermartingale argument
as in \cite{hajek1982hitting} implies that the hitting time $\tau_C$ admits an exponential moment:
there exist constants 
\begin{equation}
\label{eq:alpha-kappa}
\theta := \varepsilon/\nu^2,
\qquad
\kappa := \varepsilon^2/(2\nu^2).
\end{equation}
such that
\begin{equation}
\label{eq:exp-moment-tauC}
{E}_q\!\left[e^{\kappa \tau_C}\right]
\le
\exp\!\big(\theta(L(q)-B)\big),
\qquad \forall q.
\end{equation}

Consequently,
\begin{equation}
\label{eq:tail-tauC}
P_q(\tau_C>t)
\le
\exp\!\big(\theta(L(q)-B)\big)\,e^{-\kappa t},
\qquad t\ge 0.
\end{equation}

\textit{Step 2: Exponential tail of the regeneration time.}
Consider an ``attempt'' that starts when the process enters $C$ and lasts for $m$ slots:
the attempt is successful if the system empties (i.e., hits $Q(t)=0$) at some time within this window.
By Assumption~\ref{ass:finite-emptying}, each attempt succeeds with probability at least $p$, uniformly over the state in $C$.
If an attempt fails, then after $m$ steps, bounded increments imply that the state lies in the bounded set
\[
C' := \{q:\;
L(q)\le B+m\nu\}.
\]
From any state in $C'$, the return time to $C$ has an exponential moment uniformly bounded
by \eqref{eq:exp-moment-tauC}:
\begin{equation}
\label{eq:uniform-exp-moment-from-Cprime}
\sup_{q\in C'} E_q\!\left[e^{\kappa \tau_C}\right]
\le
\sup_{q\in C'} \exp\!\big(\theta(L(q)-B)\big)
=
\exp(\theta m\nu).
\end{equation}

Following along the lines of \cite{asmussen2007applied,meyn2009markov}, we derive the bound by conditioning on the success of the finite-time emptying attempt. Let $M(\eta) := \sup_{q \in C} E_q[e^{\eta \tau}]$ denote the worst-case moment generating function starting from the bounded set $C$ (note that $0 \in C$).

Consider a single attempt of duration $m$. We distinguish between two cases:

\textit{Success:} The system empties within $m$ steps. This occurs with probability at least $p$, and the duration is bounded by $m$. 

\textit{Failure:} The system fails to empty. This occurs with probability at most $1-p$. In this case, $m$ steps elapse, and the system ends in a state $q'$ inside the enlarged set $C'$. To regenerate, the system must first return to the set $C$ (taking time $\tau_C$) and then attempt to regenerate again (taking additional time $\tau'$).

Using the strong Markov property, we can bound $M(\eta)$ recursively:
\begin{equation}
    M(\eta) \le e^{\eta m} + (1-p) E\left[e^{\eta(m + \tau_C + \tau')}\right].
\end{equation}
The first term accounts for the moment on the success event (bounded conservatively by $e^{\eta m}$). For the failure term, we use the independence of the future regeneration time $\tau'$ to factor the expectation:
\begin{equation}
    E\left[e^{\eta(m + \tau_C + \tau')}\right] \leq e^{\eta m} \left(\sup_{q\in C'}E_q\left[e^{\eta \tau_C}\right] \right)M(\eta).
\end{equation}
Recall that for $\eta \le \kappa$, the function $x \mapsto x^{\eta/\kappa}$ is concave. For any $q\in C',$ applying Jensen's inequality to the drift bound derived in Step 1 yields:
\begin{equation}
    E_q\left[e^{\eta \tau_C}\right] \le \left( E_q\left[e^{\kappa \tau_C}\right] \right)^{\frac{\eta}{\kappa}} \le \exp\!\left(\frac{\eta}{\kappa}\theta m \nu\right).
\end{equation}
Substituting this back yields the linear inequality for $M(\eta)$:
\begin{equation}
    M(\eta) \le e^{\eta m} + (1-p)\,e^{\eta m}\,\exp\!\left(\frac{\eta}{\kappa}\theta m\nu\right) M(\eta).
\end{equation}
Solving for $M(\eta)$, we see that the moment is finite provided the coefficient of $M(\eta)$ is strictly less than 1. Thus, if there exists $\eta \in (0, \kappa]$ such that
\begin{equation}
    \label{eq:eta-condition}
    (1-p)\,e^{\eta m}\,\exp\!\left(\frac{\eta}{\kappa}\theta m\nu\right) < 1,
\end{equation}
then the moment is bounded by
\begin{equation}
    \label{eq:exp-moment-tau}
    E_0\!\left[e^{\eta \tau}\right] \le M(\eta)
    \le
    \frac{e^{\eta m}}{1-(1-p)\,e^{\eta m}\,\exp\!\left(\frac{\eta}{\kappa}\theta m\nu\right)}
    <\infty.
\end{equation}
Consequently, by Markov's inequality,
\begin{equation}
\label{eq:tail-tau}
P_0(\tau>t)\le E_0[e^{\eta \tau}]\,e^{-\eta t},
\qquad t\ge 0.
\end{equation}
\end{proof}

\paragraph{Quantum Implementation and Complexity.}
An algorithmic representation of the computations needed in one truncated regenerative cycle is provided in Appendix~\ref{app:quantum wireless} to make sure that all operations can be reversibly implemented.
The quantum simulation follows the same architecture as the $GI/GI/1$ case, with the
queue update logic replaced by the MaxWeight dynamics.
The primary difference is the cost of the update step.
In every time slot, the circuit must:
\begin{enumerate}
    \item Generate $K$ arrivals $A_i(t)$ and $K$ channel states $C_i(t)$ using the PRNG.
    \item Compute $K$ weights $w_i = Q_i(t) C_i(t)$ using reversible multipliers \cite{vedral1996}.
    \item Select the schedule $S(t) = \arg\max_i w_i$ using a reversible comparator tree with a
    fixed deterministic tie-breaking rule, requiring $O(K)$ reversible comparisons.
    \item Update the $K$ queue registers according to the MaxWeight dynamics.
\end{enumerate}
As in the $GI/GI/1$ case, the step unitary also updates all auxiliary registers required to
evaluate the truncated regeneration statistics (such as $N_M$ and $J_M(d)$); any additional
bookkeeping needed to implement the FIFO delay definition contributes only polynomial
overhead in the relevant register widths and is absorbed into the per-step cost.
Since all arithmetic operations (multiplication, comparison, and addition) admit reversible
implementations \cite{bennett1973}, the entire update map can be implemented as a reversible
circuit with intermediate workspace uncomputed at the end of each step.
The resulting circuit complexity per time slot scales linearly with the number of users $K$ and
polynomially with the register widths (i.e., $O(K\cdot\mathrm{poly}(B_Q,B_C))$), where $B_Q$
and $B_C$ denote the bit-widths of the queue-length and channel-state registers, respectively.
Consequently, the total quantum circuit complexity $T_{\mathrm{QAE}}$ scales linearly with $K$
and linearly with the truncation horizon $M$.
Given Theorem~\ref{thm:maxweight-tail}, $M$ can be chosen to grow only logarithmically in
$1/\epsilon_{\mathrm{tot}}$, where $\epsilon_{\mathrm{tot}}$ is the target additive accuracy.
This ensures that the circuit depth remains manageable while quantum amplitude estimation
provides a quadratic speedup in sample complexity over classical Monte Carlo.

\section{Load Balancing in Multi-Server Systems}
\label{sec:load-balancing}

We now consider a system of $K$ identical servers with general service time distributions with mean $1/\mu$ and Poisson arrival process with rate $\lambda$, where jobs are routed to servers according to the well-studied Join-the-Shortest (JSQ) rule. The goal is to estimate the stationary probability that a job experiences a response time (difference between departure and arrival times) greater than or equal to $d.$
The JSQ model poses a few additional challenges not seen in the other models we considered earlier:
\begin{enumerate}
    \item The inter-arrival times are unbounded and the service times could also be, so representation in terms of finite qubit registers is a challenge. We have handled this issue for the $GI/GI/1$ queue in Appendix~\ref{app:unbounded}, but multi-server model with load balancing does not have a Lindley-type recursion, making the problem more challenging.
    \item If we clip the arrival times to a bounded value, then the ``empty system'' state  is no longer a regeneration point. To handle this, we use Nummelin splitting \cite{nummelin1978splitting,henderson2001regenerative}.
    \item Unlike the wireless model, truncating the length of the regeneration cycle is not sufficient to get a finite-depth quantum circuit since the number of arrivals in a fixed time interval can be unbounded. Instead, we truncate the number of arrivals in a regeneration cycle.
\end{enumerate}

\subsection{The Clipped Surrogate Model}

To enable quantum representation, we introduce a clipping threshold $B > 0$ such that the clipped $n^{\rm th}$ inter-arrival time $A^{(B)}_n = \min\{A_n, B\}$ and clipped service time of that arrival $S^{(B)}_n = \min\{S_n, B\}$, where $A_n$ and $S_n$ are the corresponding quantities in the original system. Such a coupling of the original process and the clipped process is not required to implement quantum simulation, only the distributions should correspond to the distributions of the above clipped random variables, but the coupling interpretation will be useful later. The system still continues to operate under JSQ. We now state and prove a Lyapunov drift result for the clipped system establishing exponential tails for its regeneration times, which will be helpful later.

\begin{theorem}
\label{thm:jsq-drift-exp}
Let the state of the \emph{clipped surrogate} system at time $t$ be
\[
X(t) := \left(\mathbf{Q}(t);\, U_0(t), U_1(t), \dots, U_{K}(t)\right),
\]
where $\mathbf{Q}(t)$ is the vector of queue lengths and
$\mathbf{U}(t) \in [0,B]^{K+1}$ are the ages of the clipped arrival and service renewal processes.
Assume the original system is stable, i.e., $\lambda<K\mu$, and choose $B$ sufficiently large so that
the clipped renewal rates satisfy
\[
\lambda_B < K\mu_B,
\qquad
\lambda_B := \frac{1}{\E[A^{(B)}]},\quad
\mu_B := \frac{1}{\E[S^{(B)}]},
\]
where $A^{(B)}$ and $S^{(B)}$ denote the clipped interarrival and service times, respectively. Further, assume that there exists
$\theta_0>0$ such that
\begin{equation}
\label{eq:laplace-cond}
\E\!\left[e^{-\theta_0 S}\right] < 1.
\end{equation}

Then there exist $H>0$, a Lyapunov function $W:\mathsf{X}\to\R_{\ge 0}$,
constants $\eta>0$, $b<\infty$, $\theta_0>0$, and a compact set $\mathcal{K}$
such that for all $x$,
\begin{equation}
\label{eq:drift-W}
\E\!\left[ W(X(t+H)) - W(X(t)) \,\middle|\, X(t)=x\right]
\le -\eta 1_{\{x\notin\mathcal K\}} + b\, 1_{\{x\in\mathcal K\}},
\end{equation}
and such that the $H$-step increments of $W$ admit a uniform exponential moment:
\begin{equation}
\label{eq:mgf-W}
\sup_x
\E\!\left[
\exp\!\big(\theta_0\,|W(X(t+H))-W(X(t))|\big)
\,\middle|\,
X(t)=x
\right]
< \infty.
\end{equation}
\end{theorem}

\begin{proof}
We define the quadratic Lyapunov function $V(\mathbf Q)=\|\mathbf Q\|_2^2$.
Since the instantaneous generator drift depends on fluctuating hazard rates through the ages,
we analyze the drift over a fixed time horizon $H$. JSQ stability proofs exist in the literature for renewal arrivals and i.i.d service time distributions in the literature (see, for example, \cite{foss1998stability}), using fluid limits. Here, we need something stronger; however, since the general idea is similar, we skip over some details.

Condition on $X(t)=x$. Let $N_{A_i}$ and $N_{D_i}$ denote the number of arrivals and departures
at queue $i$ during $(t,t+H]$.
Let $N_A:=\sum_{i=1}^K N_{A_i}$ be the total number of arrivals in $(t,t+H]$, and let
$N_{S,i}$ denote the number of \emph{potential} service completions at server $i$, i.e., the number of departures if the queue were fully backlogged in $[t,t+H).$

\paragraph{Step 1: Controlling the increments.}
Note that
\begin{equation}
\label{eq:deltaV}
\Delta V
= \sum_{i=1}^K 2Q_i(t)(N_{A_i}-N_{D_i}) + \sum_{i=1}^K (\Delta Q_i)^2, \qquad\mathrm{and}
\end{equation}
\[
\sum_{i=1}^K (\Delta Q_i)^2
\le
K\Big(N_A+\max_{1\le i\le K}N_{S,i}\Big)^2 .
\]
Because inter-event times and service times are non-zero with positive probability and bounded, and the ages lie in $[0,B]$,
renewal counting processes over a fixed horizon admit uniform exponential moments \cite{asmussen2007applied}.
Hence there exists $\theta_1>0$ and $C_1(H)<\infty$ such that
\begin{equation}
\label{eq:mgf-counts}
\sup_x
\E\!\left[
\exp\!\big(\theta_1(N_A+\sum_{i=1}^K N_{S,i})\big)
\,\middle|\,
X(t)=x
\right]
\le C_1(H).
\end{equation}

\paragraph{Step 2: Uniform renewal bounds for means.}
Fix $\epsilon>0.$
Choose $H$ and $Q^*$ sufficiently large so that if $Q_i(t)\geq Q^*$
\begin{equation}
\label{eq:mean-renewal}
\E[N_A\mid x]\le (\lambda_B+\epsilon)H,
\qquad
\E[N_{D,i}\mid x]\ge (\mu_B-\epsilon)H,
\quad \forall x,i .
\end{equation}
It is clear that such a choice is possible for the arrivals from standard renewal-theoretic arguments. To ensure a similar bound for the departures from a queue, we require that the queue length be large ensuring the server is busy with high probability and allow the application of renewal theory ideas. 

\paragraph{Step 3: JSQ bound for arrivals.}
Let $s_1<\cdots<s_{N_A}$ be the arrival epochs in $(t,t+H]$.
Under JSQ routing,
\[
\sum_{i=1}^K Q_i(t)N_{A_i}
\le \sum_{\ell=1}^{N_A} Q_{\min}(s_\ell^-),
\qquad
Q_{\min}(u):=\min_i Q_i(u).
\]
Each arrival increases $Q_{\min}$ by at most $1$, hence
$Q_{\min}(s_\ell^-)\le Q_{\min}(t)+(\ell-1)$ and therefore
\[
\sum_{\ell=1}^{N_A} Q_{\min}(s_\ell^-)
\le N_A Q_{\min}(t)+\tfrac12 N_A(N_A-1).
\]
Taking conditional expectations gives
\begin{equation}
\label{eq:arrival-jsq}
\E\!\left[\sum_{i=1}^K 2Q_i(t)N_{A_i} \mid x\right]
\le
2Q_{\min}(t)\E[N_A\mid x]+\E[N_A^2\mid x].
\end{equation}

\paragraph{Step 4: Negative drift for $V$.} $\ $

\emph{Case 1: $Q_{\min}(t)\ge Q^*$.}
Then,
\[
\E\!\left[\sum_i 2Q_i(t)N_{D_i}\mid x\right]
\ge 2(\mu_B-\epsilon)H \|\mathbf Q(t)\|_1.
\]
Combining with \eqref{eq:arrival-jsq}, \eqref{eq:mean-renewal}, and \eqref{eq:deltaV} yields
\[
\E[\Delta V\mid x]
\le
2H\|\mathbf Q(t)\|_1
\Big(\tfrac{\lambda_B+\epsilon}{K}-(\mu_B-\epsilon)\Big)+C,
\]
where the coefficient of $||\mathbf{Q}||_1$ is strictly negative if $\epsilon$ is sufficiently small since $\lambda_B<K\mu_B$.

\emph{Case 2: $Q_{\min}(t)<Q^*$.}
Then \eqref{eq:arrival-jsq} yields a bounded arrival contribution. Further, large queues still satisfy \eqref{eq:mean-renewal} and if $||\mathbf{Q}||_1$ is sufficiently large compared to $KQ^*,$ we have
\[
\E[\Delta V\mid x]\le -c\,\|\mathbf Q(t)\|_1+C'
\]
for constants $c>0$, $C'<\infty$.

Thus there exist $\alpha>0$ and $C''<\infty$ such that
\begin{equation}
\label{eq:V-drift}
\E[\Delta V\mid x]\le -\alpha\|\mathbf Q(t)\|_1+C'' .
\end{equation}

\paragraph{Step 5: Passage to $W$ and exponential increment control.}
Define
\[
W(x):=\sqrt{1+\|\mathbf Q\|_2^2}.
\]
Since $\sqrt{1+\|\mathbf y\|_2^2}$ is $1$-Lipschitz,
$
|W(X(t+H))-W(X(t))|
\le \|H\mathbf Q\|_2 .
$
By \eqref{eq:mgf-counts}, (\ref{eq:laplace-cond}) and the bound
$\|H\mathbf Q\|_2\le N_A+\sum_i N_{S,i}$,
there exists $\theta_0\in(0,\theta_1]$ such that \cite{asmussen2007applied}
\[
\sup_x
\E\!\left[
\exp\!\big(\theta_0|W(X(t+H))-W(X(t))|\big)
\,\middle|\,
X(t)=x
\right]
<\infty,
\]
establishing \eqref{eq:mgf-W}.

Finally, combining \eqref{eq:V-drift} with the concavity of $y\mapsto\sqrt{1+y}$
yields the drift condition \eqref{eq:drift-W}.
\end{proof}

\subsection{Nummelin Splitting at the Empty Queue State}
\label{subsec:splitting}

Since the system operates with clipped inter-arrival times, the empty queue state is not a standard regeneration point because the future evolution depends on the time since last arrival $U_0(t)$. To estimate steady-state expectations, we employ the method of \textit{Nummelin splitting} to create independent cycles at the empty state~\cite{nummelin1978splitting,meyn1993stability,henderson2001regenerative}.
While Theorem~\ref{thm:jsq-drift-exp} guarantees that the system returns to a compact set $\mathcal{K}$ exponentially fast, one has to also ensure that the empty queue state is reached. We make the following assumption:
\begin{assumption}[Uniform reachability of the empty state]
\label{ass:reach0}
There exist constants $T_0>0$ and $p_0\in(0,1)$ such that for every state $x\in \mathcal{K}$,
\[
\mathbb{P}_x\!\left(\exists\, t\in[0,T_0]\ \text{such that}\ \mathbf{Q}(t)=\mathbf{0}\right)
\ \ge\ p_0 .
\]
\end{assumption}
Since $\mathcal{K}$ is compact, there exists a finite bound $n_\mathcal{K}$ on the total number of jobs
present whenever $X(t)\in \mathcal{K}$. Moreover, for the clipped exponential inter-arrival law
$A^{(B)}=\min(A,B)$ with $A\sim\mathrm{Exp}(\lambda)$, the event $\{A^{(B)}=B\}$ occurs with
probability $e^{-\lambda B}>0$, creating an arrival-free interval of length $B$ after an arrival.
Thus, even if the next arrival occurs almost immediately (when the arrival age is close to $B$),
after the empty queue state is reached, the subsequent inter-arrival time equals $B$ with non-zero probability.
If, in addition, the service-time distribution assigns positive probability to sufficiently fast
service completions, then with positive probability all jobs present in $\mathcal{K}$ can depart within an
interval of length $B$, thus satisfying Assumption~\ref{ass:reach0}.

\paragraph{Splitting and Minorization.}
Recall that the inter-arrival times are exponentially distributed with rate
$\lambda$, truncated at $B$, i.e.,
$A^{(B)}=\min(A,B)$ with $A\sim\mathrm{Exp}(\lambda)$.

Fix a constant $\varepsilon\in(0,B)$, and define the probability density $h$ on
$[0,\varepsilon]$ by
\[
h(y)=\frac{1}{\varepsilon}1_{[0,\varepsilon]}(y).
\]

Consider an empty-state visit with arrival age $u=U_0$.
For $u\in[0,B-\varepsilon]$, the residual time to the next arrival
$Y=A^{(B)}-u$ has a transition kernel $P_A(\cdot\mid u)$ whose absolutely
continuous part has density
\[
g_A(y\mid u)=\lambda e^{-\lambda y}, \qquad y\in[0,B-u),
\]
and which also places an atom of mass $e^{-\lambda(B-u)}$ at $y=B-u$.
In particular, for all $y\in[0,\varepsilon]$,
\[
g_A(y\mid u)\ge \lambda e^{-\lambda\varepsilon}.
\]
Consequently, for all $u\in[0,B-\varepsilon]$, the kernel $P_A(\cdot\mid u)$
admits the minorization
\[
P_A(\cdot\mid u)\ \ge\ \delta\,\phi(\cdot),
\qquad
\delta:=\varepsilon\lambda e^{-\lambda\varepsilon},
\]
where $\phi$ is the probability measure with density $h$.

Whenever the simulation enters the empty state $\mathbf 0$, we proceed as follows.
If the arrival age satisfies $U_0\le B-\varepsilon$, we perform a Nummelin
regeneration test; otherwise, no splitting is attempted and the simulation
continues until the next visit to the empty state.
When the test is performed, generate a Bernoulli random variable $Z$ with success
probability $\delta$:
\begin{itemize}
\item If $Z=1$, a regeneration event is declared and the next inter-arrival time
is sampled from $h$, independently of the past; since the system is empty, all
subsequent service times are drawn afresh from the service-time distribution,
initiating an independent regeneration cycle.
\item If $Z=0$, the next inter-arrival time is sampled from the residual kernel
\[
q_A(\cdot\mid u)
:=\frac{P_A(\cdot\mid u)-\delta\,\phi(\cdot)}{1-\delta},
\]
and the simulation continues without regeneration.
\end{itemize}

A regeneration time is defined as the first arrival epoch following an empty-state visit
at which the Bernoulli test succeeds.
Let $\tau^{(B)}$ be the length of the regeneration cycle of the clipped system. Then, from Theorem~\ref{thm:jsq-drift-exp}, one can show that $\tau^{(B)}$ has exponential tails using \cite{hajek1982hitting,meyn1993stability}, and can get explicit constants as in the proof of Theorem~\ref{thm:maxweight-tail}.

\begin{remark}
Nummelin splitting as a constructive simulation tool was introduced in \cite{henderson2001regenerative}, where its efficacy for variance estimation was demonstrated in general state-space chains. Our application differs in objective: the truncation of inter-arrival times---necessary to ensure finite quantum register width---destroys the memoryless property of the arrival process. Consequently, we are compelled to use splitting to induce artificial regeneration cycles amenable to finite-depth quantum simulation. While Nummelin splitting is standard for non-Poisson arrivals, the key observation here is that it becomes necessary even for Poisson processes once we approximate them with bounded inter-arrival times.
\end{remark}

\begin{remark}
One could ask whether clipping is theoretically necessary, given that any digital implementation inherently relies on finite precision. However, explicit clipping could offer a potential advantage in quantum resource management. Standard floating-point representations allocate fixed, large bit-widths (e.g., $64$ bits) to accommodate a vast dynamic range. By explicitly introducing a clipping threshold $B$ derived from the tail bounds of the service and inter-arrival distributions, it may be possible to limit the register width to a smaller size. This approach could allow for tailoring quantum register widths to the minimum number of qubits necessary---after accounting for required numerical precision---rather than defaulting to the standard widths typically used in classical computing. Such domain-specific clipping may yield savings in both qubit count and circuit depth without sacrificing statistical validity.
\end{remark}

\subsection{Arrival Truncation in a Regenerative Cycle}
\label{subsec:jsq-truncation}

Let 
$J_\tau(d)$ be the number of jobs whose realized response times are at
least $d$ in the regeneration cycle of the original system model without clipping. We use the same Nummelin splitting mechanism for the original system as well since the minorization also works for exponential inter-arrival times. 
Let $N_A(\tau)$ denote
the number of arrivals in that regeneration cycle.
In the clipped surrogate system, let $\tau^{(B)}$ denote the Nummelin regeneration
cycle length defined in Section~\ref{subsec:splitting}, and let
$J_{\tau^{(B)}}^{(B)}(d)$ denote the corresponding full-cycle violation count in
the clipped system.

To obtain a finite circuit depth, we truncate the number of arrivals in the clipped system to
some $R_A$. Starting from a regeneration time of the clipped system, simulate
clipped JSQ dynamics until the earlier of:
(i) completion of the clipped regeneration cycle at time $\tau^{(B)}$, or
(ii) the epoch of the $R_A$-th arrival in the cycle.
If $\tau^{(B)}$ occurs first, set $J_{R_A}(d):=J_{\tau^{(B)}}^{(B)}(d)$.
Otherwise, suppress further arrivals and continue simulating service completions
until all jobs among these first $R_A$ arrivals have departed; define $J_{R_A}(d)$
as the number of delay violations (delay $\ge d$) among these $R_A$ arrivals.
Under non-preemptive FIFO service at each server, delays of the first $R_A$
arrivals are unaffected by future arrivals, so this yields a well-defined statistic.
Moreover, this procedure simulates at most $2R_A$ events, yielding a fixed-depth
reversible implementation.

As in the previous sections, we will estimate $E[J_{R_A}(d)]$ using quantum simulation and quantify the error
\begin{equation}
\label{eq:jsq-numerator-decomp}
\bigl|\E[J_\tau(d)]-\E[J_{R_A}(d)]\bigr|
\le
\bigl|\E[J_\tau(d)]-\E[J_{\tau^{(B)}}^{(B)}(d)]\bigr|
+
\bigl|\E[J_{\tau^{(B)}}^{(B)}(d)]-\E[J_{R_A}(d)]\bigr|.
\end{equation}
Once this error is bounded, we can estimate the probability of the response time exceeding $d$ by estimating $E[N_A(\tau)]$ using classical simulation.
We next bound the two terms on the right-hand side of
\eqref{eq:jsq-numerator-decomp}.

\begin{theorem}[Truncation bias]
\label{thm:jsq-arrivalcap-bias}
For any $d\ge 0$ and any $R_A\ge 1$,
\begin{equation}
\label{eq:jsq-arrivalcap-bias}
\bigl|\E[J_{\tau^{(B)}}^{(B)}(d)]-\E[J_{R_A}(d)]\bigr|
\le
\E\!\left[N_A(\tau^{(B)})\, 1\{N_A(\tau^{(B)})>R_A\}\right],
\end{equation}
where $N_A(\tau^{(B)})$ is the number of arrivals in one clipped regeneration cycle.
Moreover, for any $M>0$,
\begin{equation}
\label{eq:jsq-arrivalcap-bias-decomp}
\E\!\left[N_A(\tau^{(B)})\, 1_{N_A(\tau^{(B)})>R_A}\right]
\le
\E\!\left[N_A(\tau^{(B)})\,1_{\tau^{(B)}>M}\right]
+
\E\!\left[N_A(M)\,1_{N_A(M)>R_A}\right],
\end{equation}
where $N_A(M)$ is the renewal count of the clipped inter-arrival process over $[0,M]$.
\end{theorem}

\begin{proof}
Couple the full-cycle clipped simulation and the clipped and truncated
simulation using the same sample path up to the epoch of the $R_A$-th arrival (or
until the cycle ends). On $\{N_A(\tau^{(B)})\le R_A\}$ the truncation is not triggered, so
the procedures coincide and
$J_{R_A}(d)=J_{\tau^{(B)}}^{(B)}(d)$.
On $\{N_A(\tau^{(B)})>R_A\}$, the truncated statistic counts response time violations only among the
first $R_A$ arrivals, hence
$|J_{\tau^{(B)}}^{(B)}(d)-J_{R_A}(d)|\le N_A(\tau^{(B)})$ pointwise.
Taking expectations gives \eqref{eq:jsq-arrivalcap-bias}.

For \eqref{eq:jsq-arrivalcap-bias-decomp}, note that
$\{N_A(\tau^{(B)})>R_A\}\subseteq \{\tau^{(B)}>M\}\cup\{N_A(M)>R_A, \tau^{(B)}\leq M\}$, and multiply
by $N_A(\tau^{(B)})$ and take expectations.
\end{proof}

To compare $J_\tau(d)$ in the original system to $J_{\tau^{(B)}}^{(B)}(d)$ in the
clipped system, we use the natural coupling described in Section~4.1: on the same
probability space, define $A_n^{(B)}=\min\{A_n,B\}$ and $S_n^{(B)}=\min\{S_n,B\}$. We also couple the random variables used for Nummelin splitting.

Let $\mathcal{E}_B$ denote the event that clipping \emph{ever matters} during a
cycle, i.e., at least one inter-arrival or service time used in the original cycle
exceeds $B$. On $\mathcal{E}_B^c$, the original and clipped sample paths coincide
(up to the corresponding regeneration cycle endpoints), hence the cycle-level delay
violation counts coincide. Therefore, the difference is supported on $\mathcal{E}_B$.

\begin{theorem}[Clipping bias]
\label{thm:jsq-clipping-bias}
Assume the coupling $A_n^{(B)}=\min\{A_n,B\}$ and $S_n^{(B)}=\min\{S_n,B\}$ is used, along with coupled Nummelin randomization. Let $\tilde{N}_A(t)$ denote the arrival counting process of the original (unclipped) system, and $N_A(t)$ denote that of the clipped system.
Then for any $d\ge 0$,
\begin{equation}
\label{eq:jsq-clipping-bias}
\bigl|{E}[J_\tau(d)]-{E}[J_{\tau^{(B)}}^{(B)}(d)]\bigr|
\le
{E}\!\left[\tilde{N}_A(\tau)\,{1}_{\mathcal{E}_B}\right]
+
{E}\!\left[N_A(\tau^{(B)})\,{1}_{\mathcal{E}_B}\right].
\end{equation}
\end{theorem}

\begin{proof}
Under the coupling, on $\mathcal{E}_B^c$ the inter-arrival and service times used by
the original and clipped systems coincide throughout the cycle. Since the sample path includes the random variables used for Nummelin splitting, the regeneration times also coincide, i.e., $\tau = \tau^{(B)}$. Consequently, the arrival counts coincide ($\tilde{N}_A(\tau) = N_A(\tau^{(B)})$) and the delay violation counts agree. Thus
\[
|J_\tau(d)-J_{\tau^{(B)}}^{(B)}(d)|
=
|J_\tau(d)-J_{\tau^{(B)}}^{(B)}(d)|\,{1}_{\mathcal{E}_B}.
\]
Moreover, $J_\tau(d)\le \tilde{N}_A(\tau)$ and $J_{\tau^{(B)}}^{(B)}(d)\le N_A(\tau^{(B)})$.
Therefore,
\[
|J_\tau(d)-J_{\tau^{(B)}}^{(B)}(d)|
\le
\bigl(\tilde{N}_A(\tau)+N_A(\tau^{(B)})\bigr){1}_{\mathcal{E}_B},
\]
and taking expectations yields \eqref{eq:jsq-clipping-bias}.
\end{proof}

In Appendix~\ref{app:approx-error-details}, we show that the bounds in the previous two theorems decay exponentially fast in $R_A$ and, if the service-time distribution has an exponential tail, then also exponentially fast in $B;$ in which case, $B$ and $R_A,$ which determine quantum complexity, can be chosen to be reasonable values for a given target error.

\subsection{Quantum Complexity and Resource Scaling}
We formulate the simulation as a reversible function $f(\omega)$ acting on a finite seed $\omega$, utilizing Bennett's strategy~\cite{bennett1973,bennett1989} to manage history; see Algorithm~\ref{alg:qdes_mgk_latency} in Appendix~\ref{app:JSQ}. The resource scaling is determined by the number of servers $K$, the arrival horizon $R_A$, and the register width of $O(\log R_A )$ qubits.
\begin{enumerate}
    \item \textit{Qubit Complexity:} The qubit count is dominated by the history tape required for uncomputation. At each of the $R_A$ simulation steps, the circuit updates \textit{all} $K$ residual service time registers via arithmetic subtractions (for clock advancement) and updates the queue length counters. To ensure reversibility, the ``borrow'' bits from these arithmetic operations and the previous states of the circular timestamp buffers must be preserved. Consequently, the total qubit requirement scales as $\tilde{O}(K \cdot R_A)$, where $\tilde{O}$ suppresses logarithmic factors for register width.
    \item \textit{Circuit Complexity:} In each simulation step, the circuit performs two dominant $O(K)$ operations: computing the global time step $\delta = \min(T_A, R_1, \dots, R_K)$, and identifying the shortest queue $j^* = \text{argmin}(Q_j)$ for JSQ routing. To implement logical early stopping within a fixed-depth circuit (since the simulation horizon is probabilistic), a global flag is used. Thus, the total gate count over the horizon $R_A$ scales linearly as $O(K \cdot R_A)$.
    \item \textit{Nummelin Splitting Logic:} The splitting logic adds an overhead of $O(K)$ gates per step. While the splitting operation itself acts on constant bits, \textit{triggering} it requires verifying the global regeneration condition. This necessitates a logical conjunction across all $K$ queue and service registers, costing $O(K)$ gates per step. But this overhead is absorbed into the dominant circuit complexity of $\tilde{O}(K \cdot R_A)$.
\end{enumerate}

\section{Conclusions}
\label{sec:conc}

In this paper, we developed a framework for estimating delay tail probabilities in queueing networks using Quantum Amplitude Estimation (QAE). By reformulating regenerative simulation as a deterministic, reversible function of a finite random seed, we bridged the fundamental gap between infinite-horizon, countable-state stochastic processes and fixed-depth quantum circuits.

A key technical contribution is the explicit characterization of truncation error via Lyapunov drift analysis. For systems satisfying standard drift conditions, we showed that the bias introduced by truncating regeneration cycles decays exponentially with the horizon. This allows the truncation horizon $M$—and consequently the quantum circuit depth—to scale only logarithmically with the target inverse precision, preserving the generic quadratic speedup of QAE even for extreme rare events.

We further demonstrated the generality of this approach by extending it to complex, continuous-state systems. Specifically, for the multi-server JSQ model with general service times, we introduced a clipped surrogate model and employed Nummelin splitting to construct artificial regeneration cycles, proving that the resulting truncation bias remains controllable. We also provided rigorous bounds on the qubit and circuit complexity, showing that the history-bit overhead required for reversibility scales linearly with $M$.

As mentioned in the introduction, combining classical variance-reduction techniques (such as importance sampling) with QAE could potentially reduce the effective regeneration cycle length. Directions for future work include developing such hybrid classical-quantum estimators for NISQ devices and extending the framework to non-regenerative metrics like the Age of Information (AoI)~\cite{yates2021age,bedewy2023sampling}.

\paragraph{Acknowledgment:} The author gratefully acknowledges the use of ChatGPT and Gemini as tools for brainstorming, literature search, and editorial assistance during the preparation of this manuscript.

\newpage
\bibliographystyle{plain}
\bibliography{references}

\newpage
\appendix
\section{Unbounded Inter--Arrival and Service Times in the $GI/GI/1$ queue}
\label{app:unbounded}

In the main text, we assumed that the
inter--arrival and service times are almost surely bounded by
$A_{\max}$ and $S_{\max}$.
This assumption yields explicit exponential tail
bounds for the regeneration time $\tau$ with constants that are easy to
interpret and estimate.
In this appendix, we briefly describe how the analysis changes when the
increments are unbounded but satisfy either sub--Gaussian or sub--exponential
tail conditions.
Fix a truncation horizon $M$ and consider the truncated rare--event
count
\[
R_M := \sum_{n=1}^{\min\{\tau,M\}} {1}\{W_n \ge d\},
\qquad
\tau := \inf\{n\ge 1: W_n=0\}.
\]
Let $B>0$ be an input clipping level and define clipped variables
\[
A^{(B)} := \min\{A,B\},\qquad S^{(B)} := \min\{S,B\}
\]
such that $E(A^{(B)})-E(S^{(B)})>0.$
Let $\{W_n^{(B)}\}$ denote the waiting time sequence produced by the Lindley
recursion using the clipped inputs $\{(A_{n+1}^{(B)},S_n^{(B)})\}$, with the same
initial condition $W_0^{(B)}=0$.
Define the corresponding truncated count
\[
R_M^{(B)} := \sum_{n=1}^{\min\{\tau^{(B)},M\}} {1}\{W_n^{(B)} \ge d\},
\qquad
\tau^{(B)} := \inf\{n\ge 1: W_n^{(B)}=0\}.
\]

The quantum algorithm ultimately estimates the normalized truncated quantity
$Y = R_M^{(B)} / M$, as in the main text.
\begin{lemma}[Clipping error for truncated counts]
\label{lem:clip-general}
Assume $\{(A_n,S_n)\}_{n\ge 1}$ are i.i.d. and the clipped variables are coupled
\emph{pathwise} by $A_n^{(B)}=\min\{A_n,B\}$ and $S_n^{(B)}=\min\{S_n,B\}$. Then
\begin{equation}
\label{eq:clip-general}
\bigl|E[R_M]-E[R_M^{(B)}]\bigr|
\;\le\;
M^2\,P(A_1>B) \;+\; M^2\,P(S_1>B).
\end{equation}
\end{lemma}

\begin{proof}
Let
\[
\mathcal{E}_B := \left\{\max_{1\le n\le M} A_n \le B \ \text{and}\ \max_{1\le n\le M} S_n \le B\right\}.
\]
On $\mathcal{E}_B$, we have $A_n^{(B)}=A_n$ and $S_n^{(B)}=S_n$ for all
$1\le n\le M$, hence the two recursions coincide up to time $M$ and therefore
$R_M=R_M^{(B)}$.
Since $0\le R_M,R_M^{(B)}\le M$ always, it follows that
\[
|R_M-R_M^{(B)}|
\le M\,{1}_{\mathcal{E}_B^c}.
\]
Taking expectations and using a union bound gives
\[
E[|R_M-R_M^{(B)}|]
\le M\,P(\mathcal{E}_B^c)
\le M\sum_{n=1}^M P(A_n>B) + M\sum_{n=1}^M P(S_n>B).
\]
The result follows from the i.i.d nature of inter-arrival and service times.
\end{proof}

\subsection{Choice of $B$ under sub--Gaussian tails}

\begin{assumption}[Sub--Gaussian tails]
\label{ass:subgaussian}
There exist parameters $(\sigma_A^2,\sigma_S^2)$ and means
$(\mu_A,\mu_S)$ such that for all $t\ge 0$,
\[
P(A_1-\mu_A \ge t) \le \exp\!\left(-\frac{t^2}{2\sigma_A^2}\right),
\qquad
P(S_1-\mu_S \ge t) \le \exp\!\left(-\frac{t^2}{2\sigma_S^2}\right).
\]
\end{assumption}

\begin{corollary}[Clipping error under sub--Gaussian tails]
\label{cor:clip-subgaussian}
Under Assumption~\ref{ass:subgaussian}, for any $B\ge \max\{\mu_A,\mu_S\}$,
\[
\bigl|E[R_M]-E[R_M^{(B)}]\bigr|
\le
M^2 \exp\!\left(-\frac{(B-\mu_A)^2}{2\sigma_A^2}\right)
+
M^2 \exp\!\left(-\frac{(B-\mu_S)^2}{2\sigma_S^2}\right).
\]
In particular, for any target $\varepsilon_{\mathrm{clip}}\in(0,1)$, it suffices to choose
\begin{equation}
\label{eq:B-choice-subg}
B \;\ge\;
\max\!\left\{
\mu_A + \sigma_A\sqrt{2\log\frac{2M^2}{\varepsilon_{\mathrm{clip}}}},
\ \ 
\mu_S + \sigma_S\sqrt{2\log\frac{2M^2}{\varepsilon_{\mathrm{clip}}}}
\right\}
\end{equation}
to ensure $\bigl|E[R_M]-E[R_M^{(B)}]\bigr|\le \varepsilon_{\mathrm{clip}}$.
\end{corollary}

\begin{proof}
Combine Lemma~\ref{lem:clip-general} with the stated tail bounds and solve for
$B$ so that each term is at most $\varepsilon_{\mathrm{clip}}/2$.
\end{proof}

\subsection{Choice of $B$ under sub--exponential tails}

\begin{assumption}[Sub--exponential tails]
\label{ass:subexp}
There exist parameters $(\lambda_A,\lambda_S)>0$ and constants $(K_A,K_S)\ge 1$
such that for all $t\ge 0$,
\[
P(A_1 \ge t) \le K_A e^{-\lambda_A t},
\qquad
P(S_1 \ge t) \le K_S e^{-\lambda_S t}.
\]
\end{assumption}

\begin{corollary}[Clipping error under sub--exponential tails]
\label{cor:clip-subexp}
Under Assumption~\ref{ass:subexp}, for any $B\ge 0$,
\[
\bigl|E[R_M]-E[R_M^{(B)}]\bigr|
\le
M^2 K_A e^{-\lambda_A B} + M^2 K_S e^{-\lambda_S B}.
\]
In particular, for any target $\varepsilon_{\mathrm{clip}}\in(0,1)$, it suffices to choose
\begin{equation}
\label{eq:B-choice-subexp}
B \;\ge\;
\max\!\left\{
\frac{1}{\lambda_A}\log\frac{2M^2K_A}{\varepsilon_{\mathrm{clip}}},
\ \ 
\frac{1}{\lambda_S}\log\frac{2M^2K_S}{\varepsilon_{\mathrm{clip}}}
\right\}
\end{equation}
to ensure $\bigl|E[R_M]-E[R_M^{(B)}]\bigr|\le \varepsilon_{\mathrm{clip}}$.
\end{corollary}

\begin{proof}
Combine Lemma~\ref{lem:clip-general} with the stated tail bounds and solve for
$B$ so that each term is at most $\varepsilon_{\mathrm{clip}}/2$.
\end{proof}

\subsection{How this changes the overall error budget and resources}

With clipping included, the natural error decomposition becomes
\[
|E[R] - E[\hat R]|
\le 
\underbrace{|E[R] - E[R_M]|}_{\text{cycle truncation}}
+ \underbrace{|E[R_M] - E[R_M^{(B)}]|}_{\text{input clipping}}
+ \underbrace{|E[R_M^{(B)}] - E[\hat R]|}_{\text{QAE}}.
\]

The cycle truncation error now has to computed using concentration inequalities for unbounded random variables, but other than that it is similar to what we had in the main body of the paper. The new term is the clipping error, which can be made $\le \varepsilon_{\mathrm{clip}}$
by choosing $B$ according to \eqref{eq:B-choice-subg} or \eqref{eq:B-choice-subexp}.
This permits a fully finite-qubit implementation while keeping the overall error
within a prescribed tolerance.
The bounded range $B$ affects the coherent implementation through the
required number of bits to represent $A_n^{(B)}$ and $S_n^{(B)}$, which scales as
$O(\log B)$ (up to discretization precision).
Thus, relative to the bounded-input
setting, the qubit and gate bounds acquire at most an additional $O(\log B)$ factor.
The analysis in the main text relies on an exponential tail bound for the regeneration time $\tau$,
obtained under bounded increments via Hoeffding’s inequality.
When the inter--arrival and service
times are unbounded and are clipped at level $B$, the resulting clipped increments
\[
X_i^{(B)} := S_i^{(B)} - A_{i+1}^{(B)}
\]
are bounded almost surely in $[-B,B]$, and therefore the same argument as in Section~\ref{sec:error-trunc-norm} applies
to the clipped regeneration time $\tau^{(B)}$.
In particular, defining
\[
\Delta^{(B)} := E[A_1^{(B)}] - E[S_1^{(B)}] > 0,
\]
we obtain an exponential tail bound of the form
\[
P(\tau^{(B)} > t) \le \exp\!\left(-\beta_B t\right),
\qquad
\beta_B := \frac{2(\Delta^{(B)})^2}{(2B)^2}.
\]
Note that the choice of $B$ to ensure $\Delta^{(B)}>0$ may depend on the tails of the inter-arrival and service time distributions.
In general, one can expect $B$ to be larger for heavier-tailed distributions.
\newpage
\section{Simulation of One Truncated Regeneration Cycle for the MaxWeight Model}\label{app:quantum wireless}


\begin{algorithm}[h]
\scriptsize
\setlength{\baselineskip}{0.9\baselineskip}
\caption{\textsc{Regeneration cycle for the MaxWeight wireless model$(\omega;M,d,I)$}}
\label{alg:wireless-one-cycle}
\begin{algorithmic}[1]
\Require Seed $\omega$ (finite bitstring for a PRNG), horizon $M\in\mathbb{Z}_{\ge 1}$, delay threshold $d\in\mathbb{Z}_{\ge 0}$, queue subset $I\subseteq\{1,\dots,K\}$
\Ensure Truncated-cycle statistics $(N_M,\;J_M(d),\;T_M,\;\mathsf{trunc})$
\State Initialize PRNG state $\mathsf{PRNG}\gets \mathsf{InitPRNG}(\omega)$
\State Initialize slot counter $t\gets 0$
\State Initialize queue lengths at the regeneration state: $Q_i(0)\gets 0$ for all $i=1,\dots,K$
\State Initialize per-queue FIFO timestamp memories $\mathcal{T}_i\gets \emptyset$ for all $i$ 
\Comment{$\mathcal{T}_i$ stores arrival times of packets currently in queue $i$}
\State Initialize counters $N_M\gets 0$, $J_M(d)\gets 0$, $\mathsf{trunc}\gets 0$
\Comment{$N_M$ counts arrivals to queues in $I$ up to time $T_M$, and $J_M(d)$ counts those with delay $\ge d$ whose departures occur by $T_M$}

\While{$t < M$}
    \Comment{\textit{(A) Observe \& sample arrivals/channel for slot $t$}}
    \State Sample arrivals $A_i(t)\in\{0,1,\dots,A_{\max}\}$ for all $i$ using $\mathsf{PRNG}$
    \State Sample channel state $\mu_i(t)\in\{0,1,\dots,\mu_{\max}\}$ for all $i$ using $\mathsf{PRNG}$
    \Comment{Arrivals and channels are i.i.d.\ over $t$ and mutually independent.}

    \Comment{\textit{(B) Scheduling decision (MaxWeight) at start of slot $t$}}
    \State Choose
    \[
      S(t)\in \arg\max_{1\le i\le K} Q_i(t)\,\mu_i(t)
    \]
    using the fixed deterministic tie-breaking rule.

    \Comment{\textit{(C) Departures first}}
    \For{$i=1,\dots,K$}
        \State $D_i(t)\gets {1}\{S(t)=i\}\min\{Q_i(t),\,\mu_i(t)\}$
    \EndFor
    \State $\Delta \gets D_{S(t)}(t)$ \Comment{Number of packets departing from the scheduled queue (possibly $0$)}
    \For{$\ell=1,\dots,\Delta$}
        \State Pop the head-of-line arrival time $a \gets \mathsf{PopFront}(\mathcal{T}_{S(t)})$
        \State \If{$S(t)\in I$ \textbf{and} $t-a \ge d$}
            \State $J_M(d)\gets J_M(d)+1$
        \EndIf
    \EndFor
    \For{$i=1,\dots,K$}
        \State $Q_i(t)\gets Q_i(t) - D_i(t)$
    \EndFor

    \Comment{\textit{(D) Arrivals next}}
    \For{$i=1,\dots,K$}
        \State $Q_i(t)\gets Q_i(t) + A_i(t)$
        \For{$\ell=1,\dots,A_i(t)$}
            \State $\mathsf{PushBack}(\mathcal{T}_i,\; t)$ \Comment{Timestamp each arrival with its slot index}
        \EndFor
        \If{$i\in I$}
            \State $N_M \gets N_M + A_i(t)$
        \EndIf
    \EndFor
    \State Set $Q_i(t+1)\gets Q_i(t)$ for all $i$ \Comment{End-of-slot update}

    \Comment{\textit{(E) Check regeneration (return to all-zero queue vector)}}
    \If{$Q_i(t+1)=0$ for all $i=1,\dots,K$}
        \State $T_M \gets t+1$
        \State \Return $(N_M,\;J_M(d),\;T_M,\;\mathsf{trunc}=0)$
    \EndIf

    \State $t\gets t+1$
\EndWhile

\Comment{\textbf{Truncation at horizon $M$}}
\State $\mathsf{trunc}\gets 1$, $T_M\gets M$
\State \Return $(N_M,\;J_M(d),\;T_M,\;\mathsf{trunc}=1)$
\end{algorithmic}
\end{algorithm}

\newpage
\section{Additional Details in Section~\ref{subsec:jsq-truncation}}
\label{app:approx-error-details}

In this section, we provide further upper bounds to the expressions in (\ref{eq:jsq-arrivalcap-bias}) and (\ref{eq:jsq-clipping-bias}).
We consider (\ref{eq:jsq-arrivalcap-bias}) first. By the Cauchy-Schwarz inequality
$$E\!\left[N_A(\tau^{(B)})\, 1_{N_A(\tau^{(B)})>R_A}\right]\leq
\sqrt{E\!\left[N_A(\tau^{(B)})^2\right]P(N_A(\tau^{(B)})>R_A\})}
$$
By using a Lyapunov drift argument in Theorem \ref{thm:jsq-drift-exp}, it can be shown that \(\tau\) has an
exponential tail, i.e., there exist constants \(c_0,\gamma>0\) such that
\begin{equation}
\label{eq:tau-exp-tail}
P(\tau^{(B)}>t)\le c_0 e^{-\gamma t},\qquad t\ge 0.
\end{equation}
We now show that \(E[N_A(\tau^{(B)})^2]<\infty\).
Fix \(\alpha\in(0,1/\lambda)\). For any integer \(m\ge 1\), by the monotonicity of
\(N_A(\cdot)\) and the union bound,
\begin{equation}
\label{eq:NA-union-bound}
\{N_A(\tau)>m\}
\subseteq
\{\tau>\alpha m\}\ \cup\ \{N_A(\alpha m)>m\},
\end{equation}
hence
\begin{equation}
\label{eq:NA-tail-split}
P(N_A(\tau)>m)\le P(\tau>\alpha m)+P(N_A(\alpha m)>m).
\end{equation}

The first term decays exponentially in \(m\) by \eqref{eq:tau-exp-tail}:
\begin{equation}
\label{eq:term1}
P(\tau^{(B)}>\alpha m)\le c_0 e^{-\gamma \alpha m}.
\end{equation}
For the second term, since \(N_A(\alpha m)\sim \mathrm{Poisson}(\lambda\alpha m)\),
a standard Chernoff bound yields, for any \(s>0\),
\[
P(N_A(\alpha m)\ge m)
\le \exp\!\Big(-sm + \lambda\alpha m (e^s-1)\Big).
\]
Optimizing over \(s\) (or taking \(s=\log(1/(\lambda\alpha))>0\)) gives
\begin{equation}
\label{eq:poisson-LD}
P(N_A(\alpha m)\ge m)\le \exp\!\big(-I(\alpha)\,m\big),
\qquad
I(\alpha):=\log\frac{1}{\lambda\alpha}-\big(1-\lambda\alpha\big)>0,
\end{equation}
where \(I(\alpha)>0\) because \(\lambda\alpha<1\).
Combining \eqref{eq:NA-tail-split}--\eqref{eq:poisson-LD}, we obtain
\begin{equation}
\label{eq:NA-exp-tail}
P(N_A(\tau^{(B)})>m)\le c_0 e^{-\gamma\alpha m} + e^{-I(\alpha)m}
\le C e^{-c m}
\end{equation}
for some constants \(C<\infty\) and \(c>0\).

Finally, since \(N_A(\tau)\) is nonnegative and integer-valued,
\begin{equation}
\label{eq:second-moment-tail-sum}
E[N_A(\tau^{(B)})^2]=\sum_{m\ge 1}(2m-1)\,P(N_A(\tau)\ge m).
\end{equation}
Using \eqref{eq:NA-exp-tail} in \eqref{eq:second-moment-tail-sum} shows that the
series converges, and therefore \(E[N_A(\tau)^2]<\infty\). By following the above steps, one can also show that $P(N_A(\tau^{(B)})>R_A\})$ decays exponentially fast in $R_A.$ 

To upper bound (\ref{eq:jsq-clipping-bias}), we recall that $\mathcal{E}_B$ denotes the event that at least one inter-arrival or service time in the cycle exceeds $B$. We observe that
\begin{equation}
{1}_{\mathcal{E}_B} \le {1}_{\mathcal{E}_B}{1}_{\{\tilde{N}_A(\tau)\leq m\}} + {1}_{\{\tilde{N}_A(\tau)> m\}} \leq \bigcup_{k=1}^{m} \left( \{ A_k > B \} \cup \{ S_k > B \} \right) + {1}_{\{\tilde{N}_A(\tau)> m\}}.
\end{equation}
Taking expectations and using the union bound,
$$
P(\mathcal{E}_B) \leq m \left( P(A_1>B) + P(S_1>B) \right) + P(\tilde{N}_A(\tau)>m).
$$
Since the inter-arrival times are exponential, if we assume that the service-times are subexponential tails then the first term decays as $mC e^{-\min(\lambda, \mu_{tail}) B}$ for some constant $C, \mu_{tail}>0$.
The second term $P(\tilde{N}_A(\tau)>m)$ decays exponentially in $m$. This follows from the same logic used for the clipped system.
Thus, by choosing $m$ linear in $B$ (e.g., $m = \kappa B$), both terms in the bound for $P(\mathcal{E}_B)$ decay exponentially in $B$.
\newpage
\setstretch{0.6}

\section{Clipped, Truncated JSQ Regeneration Cyle}
\label{app:JSQ}

\begin{algorithm}[h]
\caption{A Clipped, Truncated Regeneration Cycle of JSQ}
\label{alg:qdes_mgk_latency}
\ssmall
\begin{algorithmic}[1]
\Require 
    \Statex $R_A$: Simulation horizon (total number of arrivals to inject)
    \Statex $T_{sys}$: Global simulation time register
    \Statex $\tilde{U}_{arr}$: Residual time to next arrival
    \Statex $\tilde{U}_{1}, \dots, \tilde{U}_{K}$: Residual service times for $K$ servers
    \Statex $Q_{1}, \dots, Q_{K}$: Queue length registers
    \Statex $\mathcal{M}_{1}, \dots, \mathcal{M}_{K}$: Quantum memory buffers (arrival timestamps)
    \Statex $J$: Counter for response time violations
    \Statex $d$: Threshold constant for response time tail

\State \textbf{Initialize:} $T_{sys} \gets 0$, $J \gets 0$, $n_{arr} \gets 0$, $Q_i \gets 0$, $\tilde{U}_i \gets \infty$

\State \Comment{Loop until all arrivals generated AND system is empty}
\While{$n_{arr} < R_A \lor (\exists i : Q_i > 0 \lor \tilde{U}_i \neq \infty)$}

    \State \Comment{\textbf{Step 1: Clock Advance}}
    \State $\delta \gets \min(\tilde{U}_{arr}, \tilde{U}_1, \dots, \tilde{U}_K)$
    \State $T_{sys} \gets T_{sys} + \delta$
    \State $\tilde{U}_{arr} \gets \tilde{U}_{arr} - \delta$
    \State $\tilde{U}_i \gets \tilde{U}_i - \delta \quad \forall i \in \{1, \dots, K\}$

    \State \Comment{\textbf{Step 2: Handle Arrival Event (JSQ)}}
    \If{$\tilde{U}_{arr} = 0$}
        \If{$n_{arr} < R_A$}
            \State $n_{arr} \gets n_{arr} + 1$
            \State Find index of shortest queue: $j^* \gets \text{argmin}_{j \in \{1..K\}} (Q_j)$
            \State $Q_{j^*} \gets Q_{j^*} + 1$
            \State Push current time to buffer: $\text{Push}(\mathcal{M}_{j^*}, T_{sys})$
            
            \If{$n_{arr} < R_A$}
                \State Generate next arrival time into $\tilde{U}_{arr}$
            \Else
                \State $\tilde{U}_{arr} \gets \infty$ \Comment{Stop generating arrivals (Drain Phase)}
            \EndIf
        \Else
             \State $\tilde{U}_{arr} \gets \infty$ \Comment{Safety catch for drain phase}
        \EndIf
    \EndIf

    \State \Comment{\textbf{Step 3: Handle Departure Events}}
    \For{$i = 1$ to $K$}
        \If{$\tilde{U}_i = 0$}
            \State \Comment{Calculate Response Time Metric}
            \State $t_{arrival} \gets \text{Pop}(\mathcal{M}_i)$
            \State $\tau_{resp} \gets T_{sys} - t_{arrival}$
            
            \If{$\tau_{resp} > d$}
                \State $J \gets J + 1$
            \EndIf
            
            \State \Comment{Server Update}
            \If{$Q_i > 0$}
                \State $Q_i \gets Q_i - 1$
                \State Generate service time into $\tilde{U}_i$
            \Else
                \State $\tilde{U}_i \gets \infty$ \Comment{Server becomes idle}
            \EndIf
        \EndIf
    \EndFor

\EndWhile

\State \textbf{Return} $J * R_A^{-1}$

\end{algorithmic}
\end{algorithm}

\end{document}